\documentclass[11pt]{article}
\usepackage{amsfonts}
\usepackage{mathtools}
\usepackage{amsmath}
\usepackage{amssymb}
\usepackage{graphicx}
\usepackage[colorlinks=true,linkcolor=blue,citecolor=red,plainpages=false,pdfpagelabels]
{hyperref}
\usepackage{fullpage}
\usepackage{color}

\providecommand{\U}[1]{\protect\rule{.1in}{.1in}}
\newtheorem{theorem}{Theorem}

\newtheorem{corollary}{Corollary}

\newtheorem{definition}{Definition}

\newtheorem{lemma}{Lemma}

\newtheorem{proposition}{Proposition}

\newenvironment{proof}[1][Proof]{\noindent\textbf{#1.} }{\ \rule{0.5em}{0.5em}}

\def\Tr{\operatorname{Tr}}

\def\SEP{\operatorname{SEP}}
\def\Ent{\operatorname{Ent}}
\def\PPT{\operatorname{PPT}}
\def\supp{\operatorname{supp}}

\def\T{\operatorname{T}}
\def\c{\kappa}
\def\>{\rangle}
\def\<{\langle}
\def\({\left(}
\def\){\right)}
\def\[{\left[}
\def\]{\right]}
\def\V{\Vert}

\def\id{\operatorname{id}}

\let\emptyset\varnothing

\newcommand{\ket}[1]{\left|{#1}\right\rangle}
\newcommand{\bra}[1]{\left\langle{#1}\right|}
\newcommand{\norm}[1]{\left\Vert{#1}\right\Vert}
\newcommand{\mc}[1]{\mathcal{#1}}
\newcommand{\wt}[1]{\widetilde{#1}}
\numberwithin{equation}{section}

\allowdisplaybreaks

\begin{document}
%\preprint{ }
%\title{Capacities of bipartite quantum interactions}
\title{Resource theory of entanglement for bipartite quantum channels}
\author{Stefan B\"auml\thanks{ICFO-Institut de Ciencies Fotoniques, The Barcelona Institute of Science and Technology, Av. Carl Friedrich Gauss 3, 08860 Castelldefels (Barcelona), Spain} \thanks{QuTech, Delft University of Technology, Lorentzweg 1, 2628 CJ Delft,
Netherlands} \and
Siddhartha Das\thanks{Centre for Quantum Information \& Communication (QuIC), \'{E}cole polytechnique de Bruxelles,   Universit\'{e} libre de Bruxelles, Brussels, B-1050, Belgium} \and Xin Wang\thanks{Joint Center for Quantum Information and Computer Science, University of
Maryland, College Park, Maryland 20742, USA.} \thanks{Baidu Inc., Beijing 100193, China} \and Mark M.~Wilde\thanks{Hearne Institute for Theoretical Physics, Department of Physics and Astronomy, Center for Computation and Technology, Louisiana State University, Baton Rouge, Louisiana 70803, USA}}

\date{\today}
%\startpage{1}
%\endpage{10}
\maketitle

\begin{abstract}
The traditional perspective in quantum resource theories concerns
how to use free operations to convert one resourceful quantum state to another
one. For example, a fundamental and well known question in entanglement theory
is to determine the distillable entanglement of a bipartite state, which is
equal to the maximum rate at which fresh Bell states can be distilled from
many copies of a given bipartite state by employing local operations and
classical communication for free. It is the aim of this paper to
take this kind of question to the next level, with the main question
being:\ What is the best way of using free channels to convert one
resourceful quantum channel to another? Here we focus on the the resource theory of entanglement for bipartite channels and establish several fundamental tasks and results regarding it. In particular, we establish bounds on several pertinent information processing tasks in channel entanglement theory, and we define several entanglement measures for bipartite channels, including the logarithmic negativity and the $\kappa$-entanglement. We also show that the max-Rains information of [B\"auml \textit{et al}., Physical Review Letters, 121,  250504 (2018)] has a divergence interpretation, which is helpful for simplifying the results of this earlier work.
\end{abstract}
%\volumeyear{ }
%\volumenumber{ }
%\issuenumber{ }
%\eid{ }

\section{Introduction}

Ever since the development of the resource theory of entanglement
\cite{BDSW96,HHHH09},  the investigation of quantum
resource theories has blossomed \cite{HO13,fritz_2015,KR16,RKR17,CG18}. This is due
to such a framework being a powerful conceptual approach for understanding
physical processes, while also providing the ability to apply tools developed
in one domain to another. Any given resource theory is specified by a set
of free quantum states, as well as a set of restricted free operations, which
output a free state when the input is a free state \cite{HO13,CG18}.

%This question is the most pertinent one to consider in any quantum resource
%theory, given that quantum channels are the basic constituents of
%quantum mechanics. Indeed, quantum states, unitary evolutions, measurements,
%and discarding of quantum systems are the central components of quantum theory,
%and each of them can be regarded as a certain kind of quantum channel. From this perspective, the question posed above is the most fundamental that one could ask in any quantum resource theory, and most work on quantum resource theories thus far has only been about transformations of quantum states.
%
%going from states to channels in
%entanglement theory is a necessity:\ How does a bipartite state shared between two distant laboratories arise in the
%first place? Of course, a bipartite state originates from a bipartite channel, and from this
%perspective, it is more fundamental to develop the entanglement theory of the
%bipartite channels from which bipartite states are generated.

In the well known example
of the resource theory of entanglement \cite{BDSW96,HHHH09}, the free states
are the separable, unentangled states and the free operations consist of local
operations and classical communication (LOCC). One early insight in quantum
information theory was to modify the resource theory of entanglement to become
the resource theory of non-positive partial transpose states
\cite{Rai99,Rai01}, by enlarging the set of free states to consist of the
positive partial transpose (PPT)\ states and the class of free operations to
consist of those that preserve the PPT\ states. Consequently, it is then
possible to use this modified resource theory to deepen our understanding of
the resource theory of entanglement. Inspired by this approach, the resource theory of $k$-unextendibility was recently developed, and this consistent framework ended up giving tighter bounds on non-asymptotic rates of quantum communication \cite{KDWW18}.

The traditional approach to research on quantum resource theories is to
address the following fundamental question:\ In a given resource theory, what
is the best way to use a free quantum channel to convert one quantum state to
another? For concreteness, consider the well known resource theory of
entanglement. There, one asks about using an LOCC channel to convert from one
bipartite quantum state $\rho_{AB}$ to another bipartite state $\sigma_{AB}$.
First, is the transition possible?\ Next, what is the best asymptotic rate
$R$\ at which it is possible to start from $nR$ independent copies of
$\rho_{AB}$ and convert them approximately or exactly by LOCC\ to $n$ independent copies of
$\sigma_{AB}$? Is the resource theory reversible, in the sense that one could
start from $nR$ copies of $\rho_{AB}$, convert by LOCC to $n$ copies of
$\sigma_{AB}$, and then convert back to $nR$ copies of $\rho_{AB}$? These
kinds of questions have been effectively addressed in a number of different
works on quantum information theory
\cite{BDSW96,BBPS96,N99,Rai99,Rai01,HHT01,BP08,KH13,WD16pra,WD17}, and the
earlier works can in fact be considered the starting point for the modern
approach to quantum resource theories. 

However, upon seeing the above questions, one might have a basic question that
is not addressed by the above framework:\ \textit{How is the initial bipartite
state }$\rho_{AB}$\textit{ created in the first place?} That is, how is it
that two parties, Alice and Bob, are able to share such a state between their
distant laboratories? It is of course necessary that they employ a
communication medium, such as a fiber-optic cable or a free space link modeled
as a quantum channel, in order to do so. A model for the communication medium
is given by a bipartite quantum channel \cite{BHLS03,CLL06}, which is a
four-terminal device consisting of two inputs and two outputs, with one input
and one output for Alice and one input and one output for Bob. The basic
question above motivates developing the resource theory of entanglement for
bipartite quantum channels, and the main thrust of this paper is to do so.
The paper \cite{BHLS03} initiated this direction, but there are a large number
of questions that have remained unaddressed, and now we have a number of  tools
and conceptual approachs  to address these fundamental questions
\cite{BBCW13,BW17,DBW17,BDW18,Das2018thesis,W18cost,WW18}.

Thus, the motivation for this new direction is that quantum processes (channels) are more fundamental than quantum states, in the sense that quantum states can only arise from quantum processes, and so we should shift the focus to quantifying the resourcefulness of a quantum channel. In fact, every basic constituent of quantum mechanics, including states, unitaries, measurements, and discarding of quantum systems are particular kinds of quantum channels. In this way, a general goal  is to develop complete resource theories of quantum channels \cite{LY19,LW19}, and the outcome will be a more complete understanding of entanglement, purity, magic, coherence, distinguishability, etc.~\cite{BHLS03,DDMW17,DBW17,GFWRSCW18,BDW18,Das2018thesis,TEZP19,WW18,Seddon2019,WWS19channels,LY19,LW19,WW19states}.

Specifically, in the context of the resource theory of entanglement for bipartite quantum channels, the main question that we are interested in addressing is
this:\ Given $n$ independent uses of a bipartite quantum channel
$\mathcal{N}_{A^{\prime}B^{\prime}\rightarrow AB}$ with input quantum systems
$A^{\prime}$ and $B^{\prime}$ and output systems $A$ and $B$, as well as free
LOCC, what is the best asymptotic rate $R$\ that one can achieve for a
faithful simulation of $nR$ independent uses of another bipartite quantum
channel $\mathcal{M}_{\hat{A}^{\prime}\hat{B}^{\prime}\rightarrow\hat{A}%
\hat{B}}$ with input systems $\hat{A}^{\prime}$ and $\hat{B}^{\prime}$ and
output systems $\hat{A}$ and $\hat{B}$, in the limit of large $n$?
Furthermore, we are interested in the most general notion of channel
simulation introduced recently in \cite{W18cost}, in which the simulated
channel uses can be called in a sequential manner, by the most general
verifier who can act sequentially. Note that prior work on channel simulation
\cite{BDHSW09,BCR09,BBCW13} only considered a particular notion of channel
simulation, as well as a particular kind of channel to be simulated, in which the goal is
to simulate $nR$ independent parallel uses of a point-to-point channel
$\mathcal{P}_{A\rightarrow B}$. Also, the traditional resource theory of
entanglement for states emerges as a special case of this more general resource
theory, for the case in which the bipartite channel simply traces out the
inputs of Alice and Bob and replaces them with some bipartite state $\rho
_{AB}$.

There are certainly interesting special cases of the aforementioned general
question, which already would take us beyond what is currently known:\ How
much entanglement can be distilled from $n$ independent uses of a bipartite
channel $\mathcal{N}_{A^{\prime}B^{\prime}\rightarrow AB}$ assisted by free
LOCC? How much entanglement is required to simulate $nR$ independent uses of a
bipartite channel $\mathcal{M}_{\hat{A}^{\prime}\hat{B}^{\prime}%
\rightarrow\hat{A}\hat{B}}$, such that the most stringest verifier, who
performs a sequential test, cannot distinguish the actual channel uses from
the simulation? What if the distillation or simulation is required to be
approximate or exact? How do the rates change? How does the theory change if
we allow completely PPT-preserving channels for free, as Rains \cite{Rai99,Rai01} did? What if we allow the $k$-extendible channels of \cite{KDWW18} for free instead?

\begin{figure}[ptb]
\begin{center}
\includegraphics[
width=6.6399in
]
{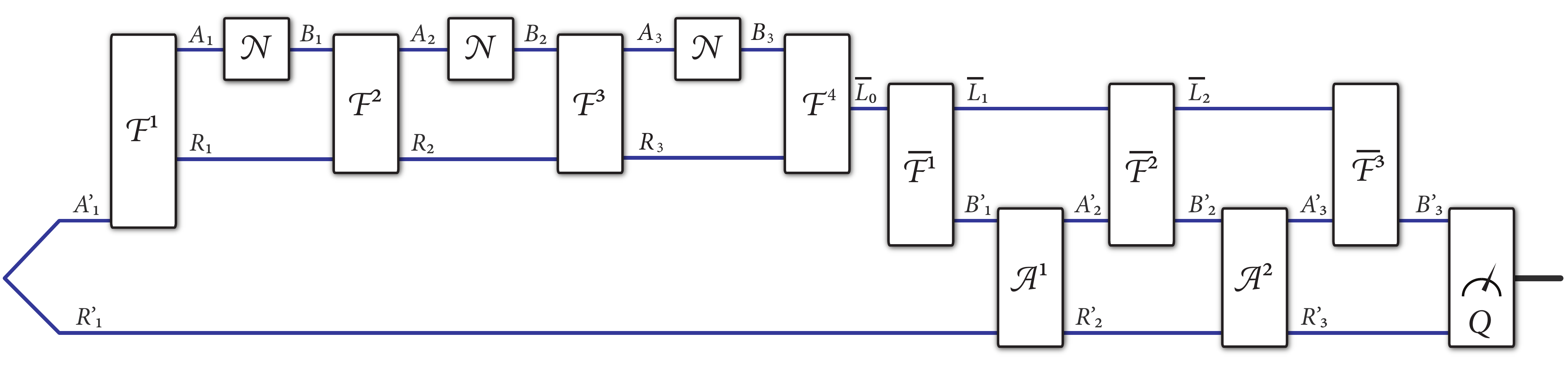}
\end{center}
\caption{The figure displays a protocol that consumes three uses of a quantum channel $\mathcal{N}_{A \to B}$ to simulate three uses of another quantum channel
$\mathcal{M}_{A' \to B'}$. Channels labeled as $\mathcal{F}$ are free in some given resource theory and can thus be consumed at no cost.  The simulation should be such that any discriminator employing an initial state on systems $R_1'A_1'$, along with adaptive channels $\mathcal{A}^1$ and $\mathcal{A}^2$ and a final measurement $Q$ on systems $R_3' B_3'$, cannot distinguish the simulation from
three uses of $\mathcal{M}_{A' \to B'}$.}
\label{fig:resource-theory-prot}%
\end{figure}

More generally, one can address these questions in general quantum resource theories. This
 constitutes a fundamental rethinking and generalization of all
of the recent work on quantum resource theories. The basic question phrased above then
becomes as follows: In a given resource theory, if $n$ independent uses of a
resourceful quantum channel $\mathcal{N}$ are available, along with the
assistance of free operations, what is the maximum possible rate $R$\ at which
one can simulate $nR$ independent uses of another resourceful channel
$\mathcal{M}$? Figure~\ref{fig:resource-theory-prot} depicts a general protocol that can accomplish
this task in any resource theory.

For the rest of the paper, we begin by giving some background in the next section. We then frame the aforementioned fundamental questions in more detail and offer solutions in some cases. The next part of the paper then proposes some entanglement measures for bipartite channels, including the logarithmic negativity, the $\kappa$-entanglement, and the generalized Rains information. We establish several fundamental properties of these measures. 

\bigskip

\textit{Note on related work:} Recently and independently of us, the resource theory of entanglement for bipartite channels was considered in \cite{GMS19}. The paper \cite{GMS19} also defined and considered some fundamental tasks in the theory, in addition to defining entanglement measures for bipartite channels, such as  logarithmic negativity and $\kappa$-entanglement.

\section{Background: States, channels, isometries, separable states, and positive partial transpose}

We begin by establishing some notation and reviewing some definitions needed in the rest of the paper. 
Let $\mc{B}(\mc{H})$ denote the 
algebra of bounded linear operators acting on a Hilbert space $\mc{H}$. Throughout this paper, we restrict our development to finite-dimensional Hilbert spaces. The
subset of $\mc{B}(\mc{H})$ 
containing all positive semi-definite operators is denoted by $\mc{B}_+(\mc{H})$. We denote the identity operator as $I$ and the identity superoperator as $\id$. The Hilbert space 
of a quantum system $A$ is denoted by $\mc{H}_A$.
The state of a quantum system $A$ is represented by a density operator $\rho_A$, which is a positive semi-definite operator with unit trace.
Let $\mc{D}(\mc{H}_A)$ denote the set of density operators, i.e., all elements $\rho_A\in \mc{B}_+(\mc{H}_A)$ such that $\Tr\{\rho_A\}=1$. The Hilbert space for a composite system $LA$ is denoted as $\mc{H}_{LA}$ where $\mc{H}_{LA}=\mc{H}_L\otimes\mc{H}_A$. The density operator of a composite system $LA$ is defined as $\rho_{LA}\in \mc{D}(\mc{H}_{LA})$, and the partial trace over $A$ gives the reduced density operator for system $L$, i.e., $\Tr_A\{\rho_{LA}\}=\rho_L$ such that $\rho_L\in \mc{D}(\mc{H}_L)$. The notation $A^n:= A_1A_2\cdots A_n$ indicates a composite system consisting of $n$ subsystems, each of which is isomorphic to the Hilbert space $\mc{H}_A$. A pure state $\psi_A$ of a system $A$ is a rank-one density operator, and we write it as $\psi_A=|\psi\>\<\psi|_A$
for $|\psi\>_A$ a unit vector in $ \mc{H}_A$. A purification of a density operator $\rho_A$ is a pure state $\psi^\rho_{EA}$
such that $\Tr_E\{\psi^\rho_{EA}\}=\rho_A$, where $E$ is called the purifying system.
The maximally mixed state is denoted by
$\pi_A := I_A / \dim(\mathcal{H}_A) \in\mc{D}\(\mc{H}_A\)$. The fidelity of $\tau,\sigma\in\mc{B}_+(\mc{H})$ is defined as $F(\tau,\sigma)=\norm{\sqrt{\tau}\sqrt{\sigma}}_1^2$ \cite{U76}, with the trace norm $\norm{X}_1=\Tr\sqrt{X^\dagger X}$ for $X\in\mc{B}(\mc{H})$.

The adjoint $\mc{M}^\dagger:\mc{B}(\mc{H}_B)\to\mc{B}(\mc{H}_A)$ of a linear map $\mc{M}:\mc{B}(\mc{H}_A)\to\mc{B}(\mc{H}_B)$ is the unique linear map such that
	\begin{equation}
	\label{eq-adjoint}
 \<Y_B,\mc{M}(X_A)\> =\<\mc{M}^\dag(Y_B),X_A\>,
	\end{equation}
	for all $X_A\in\mc{B}(\mc{H}_A)$ and $Y_B\in\mc{B}(\mc{H}_B)$,
	where $\<C,D\>=\Tr\{C^\dag D\}$ is the Hilbert-Schmidt inner product. An isometry $U:\mc{H}\to\mc{H}'$ is a
linear map such that $U^{\dag}U=I_{\mathcal{H}}$. 

The evolution of a quantum state is described by a quantum channel. A quantum channel $\mc{M}_{A\to B}$ is a completely positive, trace-preserving (CPTP) map $\mc{M}:\mc{B}_+(\mc{H}_A)\to \mc{B}_+(\mc{H}_B)$.

Let $U^\mc{M}_{A\to BE}$ denote an isometric extension of a quantum channel $\mc{M}_{A\to B}$, which by definition means that for all $\rho_A\in \mc{D}\(\mc{H}_A\)$,
\begin{equation}
\Tr_E\left\{U^\mc{M}_{A\to BE}\rho_A\left(U^\mc{M}_{A\to BE}\right)^\dagger\right\}=\mathcal{M}_{A\to B}(\rho_A) ,
\end{equation}
along with the following conditions
for $U^\mc{M}$ to be 
an isometry:
\begin{equation}
(U^\mc{M})^\dagger U^\mc{M}=I_{A}.% \qquad \text{and} \qquad U^\mc{M}(U^\mc{M})^\dagger=\Pi_{BE},
\end{equation}
Hence $U^\mc{M}(U^\mc{M})^\dagger=\Pi_{BE}$, where $\Pi_{BE}$ is a projection onto a subspace of the Hilbert space $\mc{H}_{BE}$. A complementary channel $\widehat{\mc{M}}_{A\to E}$ of $\mc{M}_{A\to B}$ is defined as
\begin{equation}
\widehat{\mc{M}}_{A\to E}(\rho_A):=\Tr_{B}\left\{U^\mc{M}_{A\to BE}\rho_A(U^\mc{M}_{A\to BE})^\dag\right\},
\end{equation}
for all $\rho_A\in \mc{D}\(\mc{H}_A\)$.

The Choi isomorphism represents a well known duality between channels and states. Let $\mc{M}_{A\to B}$ be a quantum channel, and let $\left|\Upsilon\right>_{L:A}$ denote the following maximally entangled vector:
\begin{equation}
|\Upsilon\>_{L:A}\coloneqq \sum_{i}|i\>_L|i\>_A ,
\end{equation}
where $\dim(\mc{H}_L)=\dim(\mc{H}_A)$, and $\{|i\>_L\}_i$ and $\{|i\>_A\}_i$ are fixed orthonormal bases. We extend this notation to multiple parties with a given bipartite cut as
\begin{equation}
|\Upsilon\>_{L_AL_B:AB}\coloneqq |\Upsilon\>_{L_A:A}\otimes |\Upsilon\>_{L_B:B}.
\end{equation}
The maximally entangled state $\Phi_{LA}$ is denoted as
\begin{equation}
\Phi_{LA}=\frac{1}{|A|}\ket{\Upsilon}\!\bra{\Upsilon}_{LA},
\end{equation}
where $|A|=\dim(\mc{H}_A)$.
 The Choi operator for a channel $\mc{M}_{A\to B}$ is defined as
\begin{equation}
J^\mc{M}_{LB}=(\id_L\otimes\mc{M}_{A\to B})\(|\Upsilon\>\<\Upsilon|_{LA}\),
\end{equation}
where $\id_L$ denotes the identity map on $L$. For $A'\simeq A$, the following identity holds
\begin{equation}\label{eq:choi-sim}
\<\Upsilon|_{A':L}(\rho_{SA'}\otimes J^\mc{M}_{LB})|\Upsilon\>_{A':L}=\mc{M}_{A\to B}(\rho_{SA}),
\end{equation}
where $A'\simeq A$. The above identity can be understood in terms of a post-selected variant \cite{HM04} of the quantum teleportation protocol \cite{BBC+93}. Another identity that holds is
\begin{equation}\label{eq:12}
\<\Upsilon|_{L:A} [Q_{SL}\otimes I_A]
|\Upsilon\>_{L:A}=\Tr_L\{Q_{SL}\},
\end{equation}
for an operator $Q_{SL}\in \mc{B}(\mc{H}_S\otimes\mc{H}_L)$. 

For a fixed basis $\{|i\>_B\}_i$, the partial transpose $\T_B$ on system $B$ is the following map:
\begin{equation}
\(\id_A\otimes \T_B\)(Q_{AB}) =\sum_{i,j}\(I_A\otimes |i\>\<j|_B\) Q_{AB}\( I_A\otimes |i\>\<j|_B\),\, \label{eq:PT-1}
\end{equation}
where $Q_{AB}\in\mc{B}(\mc{H}_{A}\otimes\mc{H}_{B})$. 
Further, it holds that 
\begin{equation}\label{eq:Ttrick}
\(Q_{SL}\otimes I_A\)|\Upsilon\>_{L:A}=\(T_A\(Q_{SA}\)\otimes I_L\)|\Upsilon\>_{L:A}.
\end{equation}

We note that the partial transpose is self-adjoint, i.e., $\T_B=\T^\dag_B$ and is also involutory:
\begin{equation}
\T_B\circ\T_B=I_B.
\end{equation} 
The following identity also holds:
\begin{equation}
\T_{L}(\ket{\Upsilon}\!\bra{\Upsilon}_{LA})=\T_{A}(\ket{\Upsilon}\!\bra{\Upsilon}_{LA})
\label{eq:PT-last}
\end{equation} 

Let $\SEP(A\!:\!B)$ denote the set of all separable states $\sigma_{AB}\in\mc{D}(\mc{H}_A\otimes\mc{H}_B)$, which are states that can be written as
\begin{equation}
\sigma_{AB}=\sum_{x}p(x)\omega^x_A\otimes\tau^x_B,
\end{equation}
where $p(x)$ is a probability distribution, $\omega^x_A \in \mc{D}(\mc{H}_A)$, and $\tau^x_B\in\mc{D}(\mc{H}_B)$ for all $x$. This set
is closed under the action of the partial transpose maps $\T_A$ and $\T_B$ \cite{HHH96,Per96}. Generalizing the set of separable states, we can define the set $\PPT (A\!:\!B)$ of all bipartite states $\rho_{AB}$ that remain positive after the action of the partial transpose $\T_B$. A state $\rho_{AB}\in\PPT(A\!:\!B)$ is also called a PPT (positive under partial transpose) state. We can define an even more general set of positive semi-definite operators \cite{AdMVW02} as follows:
\begin{equation}
\PPT'(A\!:\!B)\coloneqq \{\sigma_{AB}:\ \sigma_{AB}\geq 0\land \norm{\T_B(\sigma_{AB})}_1\leq 1\}. 
\end{equation} 
We then have the containments $\SEP\subset \PPT\subset \PPT' $. A bipartite quantum channel $\mc{P}_{A'B'\to AB}$ is a completely PPT-preserving channel if the map $\T_{B}\circ\mc{P}_{A'B'\to AB}\circ\T_{B'}$ is a quantum channel \cite{Rai99,Rai01,CVGG17}. A bipartite quantum channel $\mc{P}_{A'B'\to AB}$ is completely PPT-preserving if and only if its Choi state is a PPT state \cite{Rai01}, i.e., $\frac{J^{\mc{P}}_{L_AL_B:AB}}{ |L_A L_B|}\in \PPT(L_A A\!:\!BL_B)$, where
\begin{equation}
\frac{J^{\mc{P}}_{L_AL_B:AB}}{ |L_A L_B|} =  \mc{P}_{A'B'\to AB}(\Phi_{L_AA'}\otimes\Phi_{B'L_B}).
\end{equation}
Any local operations and classical communication (LOCC) channel is a completely PPT-preserving channel \cite{Rai99,Rai01}. 

\subsection{Channels with symmetry}

\label{sec:symmetry}
Consider a finite group $G$. For every $g\in G$, let $g\to U_A(g)$ and $g\to V_B(g)$ be projective unitary representations of $g$ acting on the input space $\mc{H}_A$ and the output space $\mc{H}_B$ of a quantum channel $\mc{M}_{A\to B}$, respectively. A quantum channel $\mc{M}_{A\to B}$ is covariant with respect to these representations if the following relation is satisfied \cite{Hol02,H13book}:
\begin{equation}\label{eq:cov-condition}
\mc{M}_{A\to B}\!\(U_A(g)\rho_A U_A^\dagger(g)\)=V_B(g)\mc{M}_{A\to B}(\rho_A)V_B^\dagger(g),
\end{equation}
for all $ \rho_A\in\mc{D}(\mc{H}_A)$ and $ g\in G$.
\begin{definition}[Covariant channel \cite{H13book}]\label{def:covariant}
A quantum channel is covariant if it is covariant with respect to a group $G$ which has a representation $U(g)$, for all $g\in G$, on $\mc{H}_A$ that is a unitary one-design; i.e., the map  $\frac{1}{|G|}\sum_{g\in G}U(g)(\cdot)U^\dagger(g)$ always outputs the maximally mixed state for all input states. 
\end{definition}

\begin{definition}[Teleportation-simulable \cite{BDSW96,HHH99}]\label{def:tel-sim}
A channel $\mc{M}_{A\to B}$ is teleportation-simulable with associated resource state $\omega_{L_AB}$ if for all $\rho_{A}\in\mc{D}\(\mc{H}_{A}\)$ there exists a resource state $\omega_{L_AB}\in\mc{D}\(\mc{H}_{L_AB}\)$ such that 
\begin{equation}
\mc{M}_{A\to B}\(\rho_A\)=\mc{L}_{L_AA B\to B}\(\rho_{A}\otimes\omega_{L_AB}\),
\label{eq:TP-simul}
\end{equation}
where $\mc{L}_{L_AAB\to B}$ is an LOCC channel
(a particular example of an LOCC channel is  a generalized teleportation protocol \cite{Wer01}).
\end{definition}

One can find the defining equation \eqref{eq:TP-simul} explicitly stated as \cite[Eq.~(11)]{HHH99}.
 All covariant channels, as given in  Definition~\ref{def:covariant}, are teleportation-simulable with respect to the resource state $\mathcal{M}_{A\to B}(\Phi_{L_AA})$~\cite{CDP09b}.

\begin{definition}[PPT-simulable \cite{KW17}]
A channel $\mc{M}_{A\to B}$ is PPT-simulable with associated resource state $\omega_{L_AB}$ if for all $\rho_{A}\in\mc{D}\(\mc{H}_{A}\)$ there exists a resource state $\omega_{L_AB}\in\mc{D}\(\mc{H}_{L_AB}\)$ such that 
\begin{equation}
\mc{M}_{A\to B}\(\rho_A\)=\mc{P}_{L_AA B\to B}\(\rho_{A}\otimes\omega_{L_AB}\),
\end{equation}
where $\mc{P}_{L_AAB\to B}$ is a completely PPT-preserving channel acting on $L_AA:B$, where the transposition map is with respect to the system $B$. 
\end{definition}

We note here that all of the above concepts can be generalized to bipartite channels and are helpful in the resource theory of entanglement for bipartite channels.

\section{Resource theory of entanglement for bipartite quantum channels}

To begin with, let us consider the basic ideas for the resource theory of entanglement for bipartite channels. Our specific goals are to characterize the approximate and exact entanglement costs of bipartite channels, as well as the approximate and exact distillable entanglement of bipartite channels. We can also take the free operations to be LOCC, separable, completely PPT-preserving, or $k$-extendible. These more basic problems are the basis for the more general question, as raised above, of simulating one bipartite quantum channel using another.
Let us also emphasize here that the basic questions posed can be considered in any resource theory, such as magic, purity, thermodynamics, coherence, etc. 

\subsection{Approximate and sequential entanglement cost of bipartite
quantum channels}

The first problem to consider is the entanglement cost of a
bipartite channel, and we focus first on approximate simulation in the
Shannon-theoretic sense. In \cite{W18cost}, a general definition of entanglement
cost of a single-sender, single-receiver channel was proposed, and here we extend this notion further to bipartite
channels. To this end, let $\mathcal{N}_{A^{\prime}B^{\prime}\rightarrow AB}$
denote a bipartite channel (completely positive, trace-preserving map)\ with
input systems $A^{\prime}$ and $B^{\prime}$ and output systems $A$ and $B$.
The goal is to determine the rate at which maximally entangled states are
needed to simulate $n$ uses of the bipartite channel $\mathcal{N}_{A^{\prime}B^{\prime}\rightarrow
AB}$, such that these $n$ uses could be called sequentially and thus employed
in any desired context. As discussed for the case of point-to-point channels in \cite{W18cost}, such a sequential simulation is more general and more
difficult to analyze than the prior notions of parallel channel simulation put
forward in~\cite{BBCW13}.

In more detail, let us describe what we mean by the (sequential) entanglement cost of a
bipartite channel. Fix $n,M\in\mathbb{N}$, $\varepsilon\in\left[  0,1\right]
$, and a bipartite quantum channel $\mathcal{N}_{A^{\prime}B^{\prime
}\rightarrow AB}$. We define an $(n,M,\varepsilon)$
(sequential)\ LOCC-assisted channel simulation code to consist of a maximally
entangled resource state $\Phi_{\overline{A}_{0}\overline{B}_{0}}$\ of Schmidt
rank $M$ and a set%
\begin{equation}
\{\mathcal{L}_{A_{i}^{\prime}B_{i}^{\prime}\overline{A}_{i-1}\overline
{B}_{i-1}\rightarrow A_{i}B_{i}\overline{A}_{i}\overline{B}_{i}}^{(i)}%
\}_{i=1}^{n} \label{eq:sim-prot}%
\end{equation}
of LOCC channels. Note that the systems $\overline{A}_{n}\overline{B}_{n}$ of
the final LOCC\ channel $\mathcal{L}_{A_{n}^{\prime}B_{n}^{\prime}\overline
{A}_{n-1}\overline{B}_{n-1}\rightarrow A_{n}B_{n}\overline{A}_{n}\overline
{B}_{n}}^{(n)}$ can be taken trivial without loss of generality. Alice has
access to all systems labeled by $A$, Bob has access to all systems labeled by
$B$, and they are in distant laboratories. The structure of this simulation
protocol is intended to be compatible with a discrimination strategy that can
test the actual $n$ channels versus the above simulation in a sequential way,
along the lines discussed in \cite{CDP08a,CDP09}\ and \cite{GW07,G12}. This
encompasses the parallel discrimination test, along the lines considered in
\cite{BBCW13}, as a special case.

\begin{figure}
\begin{center}
\includegraphics[
width=4.0399in
]
{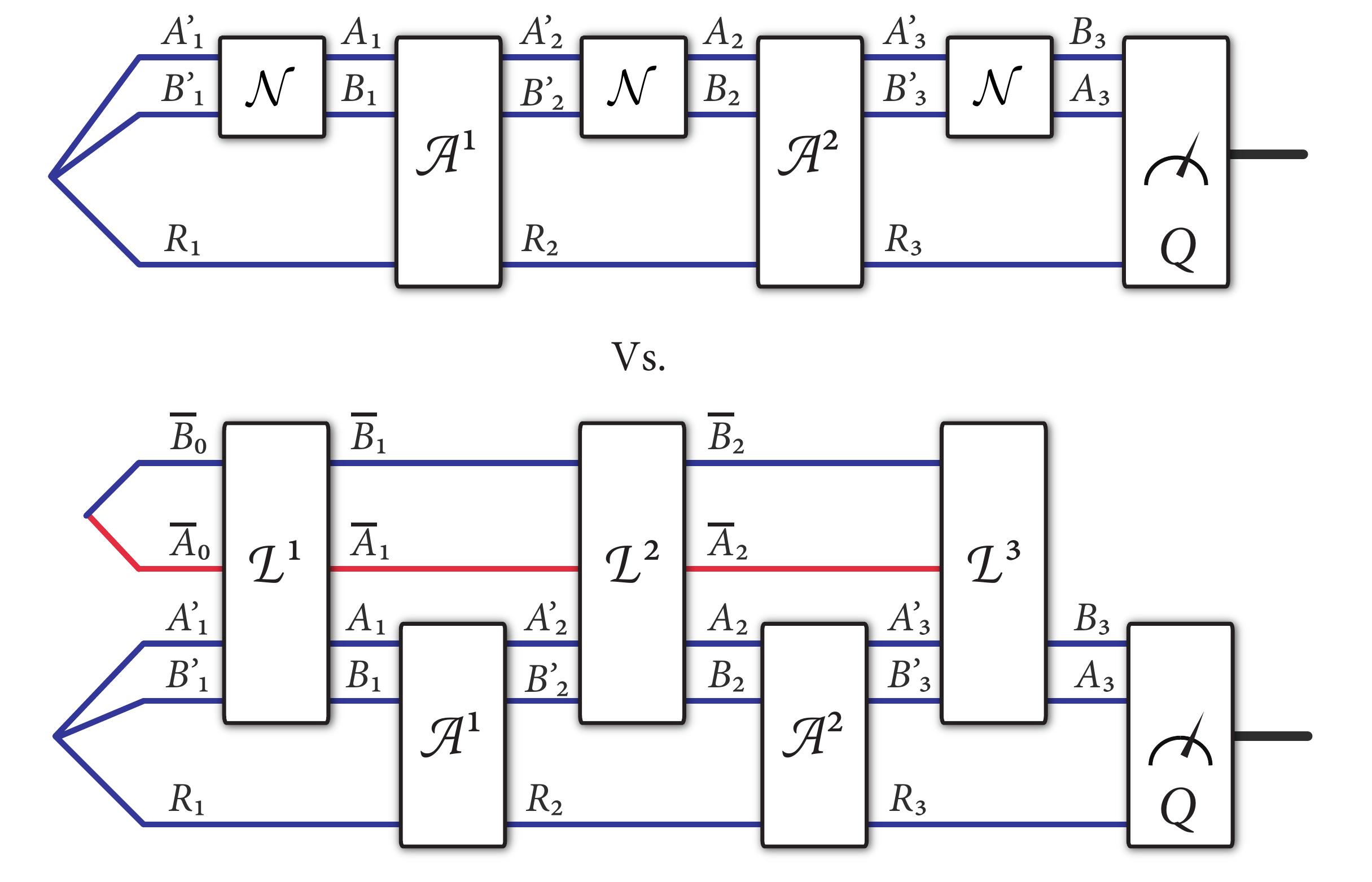}
\end{center}
\caption{The top part of the figure displays a three-round interaction between
the discriminator and the simulator in the case that the actual bipartite channel
$\mathcal{N}_{A'B'\rightarrow AB}$ is called three times. The bottom part of the
figure displays the interaction between the discriminator and the simulator in
the case that the simulation of three channel uses is called.}%
\label{fig:adaptive-prot}%
\end{figure}

A sequential discrimination strategy consists of an initial state $\rho
_{R_{1}A_{1}^{\prime}B_{1}^{\prime}}$, a set $\{\mathcal{A}_{R_{i}A_{i}%
B_{i}\rightarrow R_{i+1}A_{i+1}^{\prime}B_{i+1}^{\prime}}^{(i)}\}_{i=1}^{n-1}$
of adaptive channels, and a quantum measurement $\{Q_{R_{n}A_{n}B_{n}%
},I_{R_{n}A_{n}B_{n}}-Q_{R_{n}A_{n}B_{n}}\}$. Let the shorthand
$\{\rho,\mathcal{A},Q\}$ denote such a discrimination strategy. Note
that, in performing a discrimination strategy, the discriminator has a full
description of the bipartite channel $\mathcal{N}_{A^{\prime}B^{\prime
}\rightarrow AB}$ and the simulation protocol, which consists of
$\Phi_{\overline{A}_{0}\overline{B}_{0}}$ and the set in \eqref{eq:sim-prot}.
If this discrimination strategy is performed on the $n$ uses of the actual
channel $\mathcal{N}_{A^{\prime}B^{\prime}\rightarrow AB}$, the relevant
states involved are%
\begin{equation}
\rho_{R_{i+1}A_{i+1}^{\prime}B_{i+1}^{\prime}}:=
\mathcal{A}_{R_{i}%
A_{i}B_{i}\rightarrow R_{i+1}A_{i+1}^{\prime}B_{i+1}^{\prime}}^{(i)}%
(\rho_{R_{i}A_{i}B_{i}}),
\end{equation}
for $i\in\left\{  1,\ldots,n-1\right\}  $ and
\begin{equation}
\rho_{R_{i}A_{i}B_{i}}:=\mathcal{N}_{A_{i}^{\prime}B_{i}^{\prime
}\rightarrow A_{i}B_{i}}(\rho_{R_{i}A_{i}^{\prime}B_{i}^{\prime}}),
\end{equation}
for $i\in\left\{  1,\ldots,n\right\}  $. If this discrimination strategy is
performed on the simulation protocol discussed above, then the relevant states
involved are%
\begin{align}
\tau_{R_{1}A_{1}B_{1}\overline{A}_{1}\overline{B}_{1}}  &  :=
\mathcal{L}_{A_{1}^{\prime}B_{1}^{\prime}\overline{A}_{0}\overline{B}%
_{0}\rightarrow A_{1}B_{1}\overline{A}_{1}\overline{B}_{1}}^{(1)}(\tau
_{R_{1}A_{1}^{\prime}B_{1}^{\prime}}\otimes\Phi_{\overline{A}_{0}\overline
{B}_{0}}),\\
\tau_{R_{i+1}A_{i+1}^{\prime}B_{i+1}^{\prime}\overline{A}_{i}\overline{B}%
_{i}}  &  :=\mathcal{A}_{R_{i}A_{i}B_{i}\rightarrow R_{i+1}A_{i+1}%
^{\prime}B_{i+1}^{\prime}}^{(i)}(\tau_{R_{i}A_{i}B_{i}\overline{A}%
_{i}\overline{B}_{i}}),
\end{align}
for $i\in\left\{  1,\ldots,n-1\right\}  $, where $\tau_{R_{1}A_{1}^{\prime
}B_{1}^{\prime}}=\rho_{R_{1}A_{1}^{\prime}B_{1}^{\prime}}$, and
\begin{equation}
\tau_{R_{i}A_{i}B_{i}\overline{A}_{i}\overline{B}_{i}}:=\mathcal{L}%
_{A_{i}^{\prime}B_{i}^{\prime}\overline{A}_{i-1}\overline{B}_{i-1}\rightarrow
A_{i}B_{i}\overline{A}_{i}\overline{B}_{i}}^{(i)}(\tau_{R_{i}A_{i}^{\prime
}B_{i}^{\prime}\overline{A}_{i-1}\overline{B}_{i-1}}),
\end{equation}
for $i\in\left\{  2,\ldots,n\right\}  $. The discriminator then performs the
measurement $\{Q_{R_{n}A_{n}B_{n}},I_{R_{n}A_{n}B_{n}}-Q_{R_{n}A_{n}B_{n}}\}$
and guesses \textquotedblleft actual channel\textquotedblright\ if the outcome
is $Q_{R_{n}A_{n}B_{n}}$ and \textquotedblleft simulation\textquotedblright%
\ if the outcome is $I_{R_{n}A_{n}B_{n}}-Q_{R_{n}A_{n}B_{n}}$.
Figure~\ref{fig:adaptive-prot}\ depicts the discrimination strategy in the
case that the actual channel is called $n=3$ times and in the case that the
simulation is performed.

If the \textit{a priori} probabilities for the actual channel or simulation
are equal, then the success probability of the discriminator in distinguishing
the channels is given by%
\begin{multline}
\frac{1}{2}\left[  \operatorname{Tr}\{Q_{R_{n}A_{n}B_{n}}\rho_{R_{n}A_{n}%
B_{n}}\}+\operatorname{Tr}\{\left(  I-Q\right)_{R_{n}A_{n}B_{n}}  \tau_{R_{n}A_{n}B_{n}}\}\right] \\
  \leq\frac{1}{2}\left(  1+\frac
{1}{2}\left\Vert \rho_{R_{n}A_{n}B_{n}}-\tau_{R_{n}A_{n}B_{n}}\right\Vert
_{1}\right)  ,
\end{multline}
where the latter inequality is well known from the theory of quantum state
discrimination \cite{H69,H73,Hel76}. For this reason, we say that the $n$
calls to the actual channel $\mathcal{N}_{A^{\prime}B^{\prime}\rightarrow AB}$
are $\varepsilon$-distinguishable from the simulation if the following
condition holds for the respective final states%
\begin{equation}
\frac{1}{2}\left\Vert \rho_{R_{n}A_{n}B_{n}}-\tau_{R_{n}A_{n}B_{n}}\right\Vert
_{1}\leq\varepsilon.
\end{equation}
If this condition holds for all possible discrimination strategies
$\{\rho,\mathcal{A},Q\}$, i.e., if%
\begin{equation}
\frac{1}{2}\sup_{\left\{  \rho,\mathcal{A}\right\}  }\left\Vert \rho
_{R_{n}A_{n}B_{n}}-\tau_{R_{n}A_{n}B_{n}}\right\Vert _{1}\leq\varepsilon,
\label{eq:good-sim}%
\end{equation}
then the simulation protocol constitutes an $(n,M,\varepsilon)$ channel
simulation code. It is worthwhile to remark: If we ascribe the shorthand
$(\mathcal{N})^{n}$ for the $n$ uses of the channel and the shorthand
$(\mathcal{L})^{n}$ for the simulation, then the condition in
\eqref{eq:good-sim} can be understood in terms of the $n$-round strategy norm
of \cite{CDP08a,CDP09,G12}:%
\begin{equation}
\frac{1}{2}\left\Vert (\mathcal{N})^{n}-(\mathcal{L})^{n}\right\Vert
_{\Diamond,n}\leq\varepsilon. \label{eq:strategy-norm}%
\end{equation}

A rate $R$ is achievable for (sequential)\ bipartite channel
simulation of $\mathcal{N}$ if for all $\varepsilon\in(0,1]$, $\delta>0$, and
sufficiently large $n$, there exists an $(n,2^{n\left[  R+\delta\right]
},\varepsilon)$ (sequential)\ bipartite channel simulation code for
$\mathcal{N}$. The (sequential)\ entanglement cost $E_{C}%
(\mathcal{N})$ of the bipartite channel $\mathcal{N}$ is defined to be the infimum of all
achievable rates.

The main question here is to identify a general mathematical expression for
the entanglement cost $E_{C}(\mathcal{N})$ as defined above. This could end up
being a very difficult problem in general, but one can attack the problem in a variety of ways. Below we discuss some specific instances.

A special kind of distinguisher only employs a parallel distinguishing strategy, similar to the approach taken in prior work \cite{BBCW13}. Even this scenario has not been considered previously in the context of bipartite channels. However, in what follows, we center the discussion around sequential simulation as presented above.

As another variation, we can consider the free operations to be completely PPT-preserving channels \cite{Rai99,Rai01} rather than LOCC channels, as was done in the work of Rains \cite{Rai99,Rai01}. Since the set of completely PPT-preserving channels contains LOCCs, this approach can be useful for obtaining bounds on the entanglement cost. This approach was taken recently in \cite{WW18}, for single-sender, single-receiver channels.

\paragraph{Approximate and sequential entanglement cost for resource-seizable
bipartite channels}

First, let us discuss a special case, by supposing that the bipartite channel
has some structure, i.e., that it is bidirectional teleportation simulable as
defined in \cite{STM11,DBW17}:

\begin{definition}
[Bidirectional teleportation-simulable]\label{def:bi-tel-sim} A bipartite
channel $\mathcal{N}_{A^{\prime}B^{\prime}\rightarrow AB}$ is
teleportation-simulable with associated resource state $\theta_{L_{A}L_{B}}$
if for all input states $\rho_{AB}$ the following equality holds
\begin{equation}
\mathcal{N}_{A^{\prime}B^{\prime}\rightarrow AB}(\rho_{AB})=\mathcal{L}%
_{L_{A}ABL_{B}\rightarrow A^{\prime}B^{\prime}}(\rho_{AB}\otimes\theta
_{L_{A}L_{B}}), \label{eq:tele-sim-ch}%
\end{equation}
where $\mathcal{L}_{L_{A}ABL_{B}\rightarrow A^{\prime}B^{\prime}}$ is an LOCC
channel acting on $L_{A}A:L_{B}B$.
\end{definition}

A special kind of bidirectional teleportation-simulable channel is one that is
resource-seizable, in a sense that generalizes a similar notion put forward in
\cite{W18cost,BHKW18}.

\begin{definition}
[Resource seizable]\label{def:resource-seizable}Let $\mathcal{N}_{A^{\prime
}B^{\prime}\rightarrow AB}$ denote a bipartite channel that is
teleportation-simulable with associated resource state $\theta_{L_{A}L_{B}}$.
It is resource-seizable if there exists a separable input state $\rho
_{A_{M}A^{\prime}B^{\prime}B_{M}}$ and an LOCC\ channel $\mathcal{D}%
_{A_{M}ABB_{M}\rightarrow L_{A}L_{B}}$ such that%
\begin{equation}
\mathcal{D}_{A_{M}ABB_{M}\rightarrow L_{A}L_{B}}(\mathcal{N}_{A^{\prime
}B^{\prime}\rightarrow AB}(\rho_{A_{M}A^{\prime}B^{\prime}B_{M}}%
))=\theta_{L_{A}L_{B}}.
\end{equation}

\end{definition}

\begin{theorem}
\label{thm:e-cost-bi-res-seize}Let $\mathcal{N}_{A^{\prime
}B^{\prime}\rightarrow AB}$ denote a bipartite channel that is
teleportation-simulable and resource-seizable. Then its (sequential)
entanglement cost is equal to the entanglement cost of the underlying resource
state:%
\begin{equation}
E_{C}(\mathcal{N}_{A^{\prime}B^{\prime}\rightarrow AB})=E_{C}(\theta
_{L_{A}L_{B}}).
\end{equation}

\end{theorem}

The proof of this theorem follows along the lines of the proof of
\cite[Theorem~1]{W18cost}. To achieve the rate $E_{C}(\theta_{L_{A}L_{B}})$,
Alice and Bob use maximal entanglement at the rate $E_{C}(\theta_{L_{A}L_{B}%
})$ to make a large number $n$\ of approximate copies of the resource state
$\theta_{L_{A}L_{B}}$. Then whenever the channel simulation is needed, they
use one of the approximate copies along with the LOCC\ channel from
\eqref{eq:tele-sim-ch} to complete the simulation. Related to the observations
from \cite[Proposition~2]{W18cost}, the ability of a verifier to distinguish
the bipartite channel $\mathcal{N}_{A^{\prime}B^{\prime}\rightarrow AB}$ from
its simulation is limited by the distinguishability of the resource state
$\theta_{L_{A}L_{B}}^{\otimes n}$ from its approximation, which can be made
arbitrarily small with increasing $n$. The converse part follows by employing
the entanglement of formation, its properties, a parallel verification test,
and the resource-seizable property from Definition~\ref{def:resource-seizable}%
\ to deduce that the entanglement cost should at least be equal to
$E_{C}(\theta_{L_{A}L_{B}})$.

Particular bipartite channels that are bidirectional teleportation simulable are those that are bicovariant, as defined and identified in \cite{DBW17}. For such channels, we can conclude from Theorem~\ref{thm:e-cost-bi-res-seize} that their entanglement cost is equal to the entanglement cost of their Choi states.

\paragraph{Beyond resource-seizable channels}

It is of interest to characterize the entanglement cost for general
bipartite channels, beyond those discussed above. A successful approach in
prior work \cite{BBCW13} was to apply the quantum de Finetti theorem / reduction \cite{CKR09} to
simplify the analysis. There, the authors of \cite{BBCW13} took advantage of
permutation symmetry inherent in the channel being simulated, and the 
finding is that rather than having to test the performance of the simulation
protocol on every possible state, it is only necessary to do so for a single
universal de Finetti state, at the price of a polynomial in $n$ multiplicative factor
for the error of the simulation. However, in the asymptotic limit, this
polynomial factor is negligible and does not affect the simulation cost.

A task to consider here, as mentioned above, is to restrict the notion of simulation to be
a parallel simulation, as done in \cite{BBCW13}, in which the goal is to simulate $n$
parallel uses of the bipartite channel $\mathcal{N}_{A^{\prime}B^{\prime
}\rightarrow AB}$, i.e., to simulate $(\mathcal{N}_{A^{\prime}B^{\prime
}\rightarrow AB})^{\otimes n}$. In particular, the goal of this simplified
notion of bipartite channel simulation is to consider a simulation protocol
$\mathcal{P}_{A^{\prime n}B^{\prime n}\rightarrow A^{n}B^{n}}$ to have the
following form:%
\begin{equation}
\mathcal{P}_{A^{\prime n}B^{\prime n}\rightarrow A^{n}B^{n}}(\omega_{A^{\prime
n}B^{\prime n}}):=\mathcal{L}_{A^{\prime n}B^{\prime n}\overline{A}%
_{0}\overline{B}_{0}\rightarrow A^{n}B^{n}}(\omega_{A^{\prime n}B^{\prime n}%
}\otimes\Phi_{\overline{A}_{0}\overline{B}_{0}}),
\label{eq:berta-parallel-prot}%
\end{equation}
where $\omega_{A^{\prime n}B^{\prime n}}$ is an arbitrary input state,
$\mathcal{L}_{A^{\prime n}B^{\prime n}\overline{A}_{0}\overline{B}%
_{0}\rightarrow A^{n}B^{n}}$ is a free LOCC\ channel, and $\Phi_{\overline
{A}_{0}\overline{B}_{0}}$ is a maximally entangled resource state. For
$\varepsilon\in\left[  0,1\right]  $, the simulation is then considered
$\varepsilon$-distinguishable from $(\mathcal{N}_{A^{\prime}B^{\prime
}\rightarrow AB})^{\otimes n}$ if the following condition holds%
\begin{equation}
\frac{1}{2}\left\Vert (\mathcal{N}_{A^{\prime}B^{\prime}\rightarrow
AB})^{\otimes n}-\mathcal{P}_{A^{\prime n}B^{\prime n}\rightarrow A^{n}B^{n}%
}\right\Vert _{\Diamond}\leq\varepsilon, \label{eq:berta-sim}%
\end{equation}
where $\left\Vert \cdot\right\Vert _{\Diamond}$ denotes the diamond norm
\cite{Kit97}. The physical meaning of the above inequality is that it places a
limitation on how well any discriminator can distinguish the channel
$(\mathcal{N}_{A^{\prime}B^{\prime}\rightarrow AB})^{\otimes n}$ from the
simulation $\mathcal{P}_{A^{\prime n}B^{\prime n}\rightarrow A^{n}B^{n}}$ in a
guessing game. Such a guessing game consists of the discriminator preparing a
quantum state $\rho_{RA^{\prime n}B^{\prime n}}$, the referee picking
$(\mathcal{N}_{A^{\prime}B^{\prime}\rightarrow AB})^{\otimes n}$ or
$\mathcal{P}_{A^{\prime n}B^{\prime n}\rightarrow A^{n}B^{n}}$ at random and
then applying it to the $A^{\prime n}B^{\prime n}$ systems of $\rho
_{RA^{\prime n}B^{\prime n}}$, and the discriminator finally performing a
quantum measurement on the systems $RA^{n}B^{n}$. If the inequality in
\eqref{eq:berta-sim} holds, then the probability that the discriminator can
correctly distinguish the channel from its simulation is bounded from above by
$\frac{1}{2}\left(  1+\varepsilon\right)  $, regardless of the particular
state $\rho_{RA^{\prime n}B^{\prime n}}$ and final measurement chosen for his
distinguishing strategy \cite{Kit97,H69,H73,Hel76}. Thus, if $\varepsilon$ is
close to zero, then this probability is not much better than random guessing,
and in this case, the channels are considered nearly indistinguishable and the
simulation thus reliable.

\subsection{Exact and sequential entanglement cost of bipartite channels}

Another important scenario to consider is the \textit{exact} entanglement cost
of a bipartite channel. Here, the setting is the same as that described above,
but the goal is to incur no error whatsoever when simulating a bipartite
channel. That is, it is required that $\varepsilon=0$ in \eqref{eq:good-sim},
\eqref{eq:strategy-norm}, and \eqref{eq:berta-sim}. Even though such a change
might seem minimal, it has a dramatic effect on the theory and how one attacks
the problem. There are at least two possible ways to approach the exact case,
by allowing the free operations to be LOCC\ or completely PPT-preserving channels.

Let us first discuss the second case. In \cite{WW18},  
 the $\kappa$-entanglement of a bipartite state $\rho_{AB}$\ was defined as
follows:%
\begin{equation}
E_{\kappa}(\rho_{AB}):=\log_{2}\inf\{\operatorname{Tr}\{S_{AB}%
\}:-T_{B}(S_{AB})\leq T_{B}(\rho_{AB})\leq T_{B}(S_{AB}),\ S_{AB}%
\geq0\},\label{eq:kappa-ent-states}%
\end{equation}
where $T_{B}$ denotes the partial transpose.
As proven in \cite{WW18},
the entanglement measure $E_\kappa$ has many
desirable properties, including monotonicity under selective completely PPT-preserving
operations and additivity $E_{\kappa}(\rho_{AB}\otimes\sigma_{A^{\prime
}B^{\prime}})=E_{\kappa}(\rho_{AB})+E_{\kappa}(\sigma_{A^{\prime}B^{\prime}}%
)$. It is also efficiently computable by a semi-definite program. Furthermore,
it has an operational meaning as the exact entanglement cost of a bipartite
state $\rho_{AB}$. That is, we define the one-shot exact entanglement cost of
a bipartite state $\rho_{AB}$ as
\begin{equation}
E_{\operatorname{PPT}}^{(1,c)}(\rho_{AB}):=\inf\{\log_{2}d:\mathcal{P}%
_{\hat{A}\hat{B}\rightarrow AB}(\Phi_{\hat{A}\hat{B}}^{d})=\rho_{AB}\},
\end{equation}
where $\Phi_{\hat{A}\hat{B}}^{d}$ denotes a maximally entangled state of
Schmidt rank $d$ and $\mathcal{P}_{\hat{A}\hat{B}\rightarrow AB}$ denotes a
completely PPT-preserving channel. Then the one-shot entanglement cost of $\rho_{AB}$ is
defined as
\begin{equation}
E_{\operatorname{PPT}}^{c}(\rho_{AB}):=\lim_{n\rightarrow
\infty}\frac{1}{n}E_{\operatorname{PPT}}^{(1,c)}(\rho_{AB}^{\otimes n}),
\end{equation}
and one of
the main results of \cite{WW18} is that%
\begin{equation}
E_{\operatorname{PPT}}^{c}(\rho_{AB})=E_{\kappa}(\rho_{AB}%
).\label{eq:kappa-cost}%
\end{equation}
Thus, this represents the first time in  entanglement theory
that an entanglement measure for general bipartite states is both efficiently
computable while having an operational meaning.

Another accomplishment of \cite{WW18} was to establish that the exact entanglement
cost of a single-sender, single-receiver quantum channel $\mathcal{N}%
_{A\rightarrow B}$, as defined in the previous section but with $\varepsilon
=0$ and with free LOCC\ operations replaced by completely PPT-preserving operations, is
given by
\begin{equation}
E_{\operatorname{PPT}}(\mathcal{N}_{A\rightarrow B})=E_{\kappa}(\mathcal{N}%
_{A\rightarrow B}):=\sup_{\psi_{RA}}E_{\kappa}(\mathcal{N}_{A\rightarrow
B}(\psi_{RA})),
\end{equation}
where the optimization on the right-hand side is with respect to pure,
bipartite states $\psi_{RA}$ with system $R$ isomorphic to system $A$. The
quantity $E_{\kappa}(\mathcal{N}_{A\rightarrow B})$ is called the $\kappa
$-entanglement of a quantum channel in \cite{WW18}, where it was also shown to be
efficiently computable via a semi-definite program and to not increase under
amortization (a property stronger than additivity). It has a dual
representation, via semi-definite programming duality, as follows:%
\begin{equation}
E_{\kappa}(\mathcal{N}_{A\rightarrow B})=\log_{2}\inf\{\left\Vert
\operatorname{Tr}_{B}Q_{RB}\right\Vert _{\infty}:-T_{B}(Q_{RB})\leq
T_{B}(J_{RB}^{\mathcal{N}})\leq T_{B}(Q_{RB}),\ Q_{RB}\geq0\},
\label{eq:kappa-ent-channel}%
\end{equation}
where $J_{RB}^{\mathcal{N}}$ is the Choi operator of the channel
$\mathcal{N}_{A\rightarrow B}$. This dual representation bears an interesting
resemblance to the formula for the $\kappa$-entanglement of bipartite states
in \eqref{eq:kappa-ent-states}.

Extending the result of \cite{WW18}, we can consider  the exact
entanglement cost of a bipartite channel $\mathcal{N}_{A^{\prime}B^{\prime
}\rightarrow AB}$. In light of the above result, it is reasonable that
the exact entanglement cost should simplify so much as to lead to an
efficiently-computable and single-letter formula. At the least, a reasonable guess
for an appropriate formula for the $\kappa$-entanglement of a bipartite
channel $\mathcal{N}_{A^{\prime}B^{\prime}\rightarrow AB}$, in light of the
prior two results, is as follows:%
\begin{multline}
E_{\kappa}(\mathcal{N}_{A^{\prime}B^{\prime}\rightarrow AB})=\log_{2}%
\inf\{\left\Vert \operatorname{Tr}_{AB}Q_{R_{A}ABR_{B}}\right\Vert _{\infty
}:\label{eq:kappa-bipartite-channel}\\
-T_{BR_{B}}(Q_{R_{A}ABR_{B}})\leq T_{BR_{B}}(J_{R_{A}ABR_{B}}^{\mathcal{N}%
})\leq T_{BR_{B}}(Q_{R_{A}ABR_{B}}),\ Q_{R_{A}ABR_{B}}\geq0\},
\end{multline}
where $J_{R_{A}ABR_{B}}^{\mathcal{N}}$ is the Choi operator of the bipartite
channel $\mathcal{N}_{A^{\prime}B^{\prime}\rightarrow AB}$, with the systems
$R_{A}$ and $R_{B}$ being isomorphic to the respective channel input systems
$A^{\prime}$ and $B^{\prime}$. The above $\kappa
$-entanglement of a bipartite channel reduces to the correct formula in
\eqref{eq:kappa-ent-states}\ when the bipartite channel is equivalent to a
bipartite state $\rho_{AB}$, with its action to trace out the input systems
$A^{\prime}$ and $B^{\prime}$ and replace with the state $\rho_{AB}$. This is
because the Choi operator $J_{R_{A}ABR_{B}}^{\mathcal{N}}=I_{R_{A}}\otimes
\rho_{AB}\otimes I_{R_{B}}$ in such a case, and then the optimization above
simplifies to the formula in \eqref{eq:kappa-ent-states}. Furthermore, when
the bipartite channel $\mathcal{N}_{A^{\prime}B^{\prime}\rightarrow AB}$\ is
just a single-sender, single-receiver channel, with trivial $B^{\prime}$
system and trivial $A$ system, then $R_{B}$ and $A$ of $J_{R_{A}ABR_{B}%
}^{\mathcal{N}}$ are trivial, so that the formula above reduces to the correct
formula in \eqref{eq:kappa-ent-channel}. Later we show that this measure is a good measure of entanglement for bipartite channels, in the sense that it obeys several desirable properties.

\subsection{Approximate distillable entanglement of bipartite channels}

Given a bipartite channel $\mathcal{N}_{A^{\prime}B^{\prime}\rightarrow AB}$,
we are also interested in determining its distillable entanglement, which is a
critical component of the resource theory of entanglement for bipartite channels.

\begin{figure}
		\centering
		\includegraphics[scale=0.579]{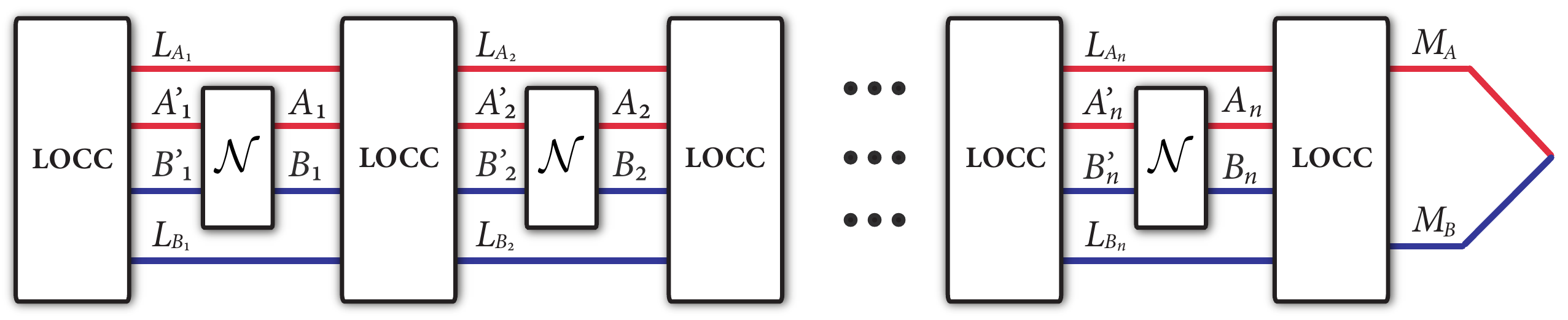}
		\caption{A protocol for LOCC-assisted entanglement distillation that uses a bipartite quantum channel $\mathcal{N}$ $n$ times. Every channel use is interleaved by an LOCC bipartite channel. The goal of such a protocol is to produce an approximate maximally entangled state in the systems $M_A$ and $M_B$, where Alice possesses system $M_A$ and Bob system $M_B$.}\label{fig:bi-q-com}
	\end{figure}

The most general protocol for distilling entanglement from a bipartite
channel, as depicted in Figure~\ref{fig:bi-q-com}, has the following form.
Alice and Bob are spatially separated, and they are allowed to undergo a
bipartite quantum channel $\mathcal{N}_{A^{\prime}B^{\prime}\rightarrow AB}$.
Alice holds systems labeled by $A^{\prime},A$ whereas Bob holds $B^{\prime}%
,B$. They begin by performing an LOCC channel $\mathcal{L}_{\emptyset
\rightarrow L_{A_{1}}A_{1}^{\prime}B_{1}^{\prime}L_{B_{1}}}^{(1)}$, which
leads to a separable state $\rho_{L_{A_{1}}A_{1}^{\prime}B_{1}^{\prime
}L_{B_{1}}}^{(1)}$, where $L_{A_{1}},L_{B_{1}}$ are finite-dimensional systems
of arbitrary size and $A_{1}^{\prime},B_{1}^{\prime}$ are input systems to the
first channel use. Alice and Bob send systems $A_{1}^{\prime}$ and
$B_{1}^{\prime}$, respectively, through the first channel use, which yields
the output state $\sigma_{L_{A_{1}}A_{1}B_{1}L_{B_{1}}}^{(1)}:=\mathcal{N}%
_{A_{1}^{\prime}B_{1}^{\prime}\rightarrow A_{1}B_{1}}(\rho_{L_{A_{1}}%
A_{1}^{\prime}B_{1}^{\prime}L_{B_{1}}}^{(1)})$. Alice and Bob then perform the
LOCC channel $\mathcal{L}_{L_{A_{1}}A_{1}B_{1}L_{B_{1}}\rightarrow L_{A_{2}%
}A_{2}^{\prime}B_{2}^{\prime}L_{B_{2}}}^{(2)}$, which leads to the state
$\rho_{L_{A_{2}}A_{2}^{\prime}B_{2}^{\prime}L_{B_{2}}}^{(2)}:=\mathcal{L}%
_{L_{A_{1}}A_{1}B_{1}L_{B_{1}}\rightarrow L_{A_{2}}A_{2}^{\prime}B_{2}%
^{\prime}L_{B_{2}}}^{(2)}(\sigma_{L_{A_{1}}A_{1}B_{1}L_{B_{1}}}^{(1)})$. Both
parties then send systems $A_{2}^{\prime},B_{2}^{\prime}$ through the second
channel use $\mathcal{N}_{A_{2}^{\prime}B_{2}^{\prime}\rightarrow A_{2}B_{2}}%
$, which yields the state
\begin{equation}
\sigma_{L_{A_{2}}A_{2}B_{2}L_{B_{2}}}%
^{(2)}:=\mathcal{N}_{A_{2}^{\prime}B_{2}^{\prime}\rightarrow A_{2}B_{2}}%
(\rho_{L_{A_{2}}A_{2}^{\prime}B_{2}^{\prime}L_{B_{2}}}^{(2)}).
\end{equation}
They iterate
this process such that the protocol makes use of the channel $n$ times. In
general, we have the following states for the $i$th use, for $i\in
\{2,3,\ldots,n\}$:
\begin{align}
\rho_{L_{A_{i}}A_{i}^{\prime}B_{i}^{\prime}L_{B_{i}}}^{(i)}  &  :=\mathcal{L}%
_{L_{A_{i-1}}A_{i-1}B_{i-1}L_{B_{i-1}}\rightarrow L_{A_{i}}A_{i}^{\prime}%
B_{i}^{\prime}L_{B_{i}}}^{(i)}(\sigma_{L_{A_{i-1}}A_{i-1}B_{i-1}L_{B_{i-1}}%
}^{(i-1)}),\\
\sigma_{L_{A_{i}}A_{i}B_{i}L_{B_{i}}}^{(i)}  &  :=\mathcal{N}_{A_{i}^{\prime
}B_{i}^{\prime}\rightarrow A_{i}B_{i}}(\rho_{L_{A_{i}}A_{i}^{\prime}%
B_{i}^{\prime}L_{B_{i}}}^{(i)}),
\end{align}
where $\mathcal{L}_{L_{A_{i-1}}A_{i-1}B_{i-1}L_{B_{i-1}}\rightarrow L_{A_{i}%
}A_{i}^{\prime}B_{i}^{\prime}L_{B_{i}}}^{(i)}$ is an LOCC channel. In the
final step of the protocol, an LOCC channel $\mathcal{L}_{L_{A_{n}}A_{n}%
B_{n}L_{B_{n}}\rightarrow M_{A}M_{B}}^{(n+1)}$ is applied, that generates the
final state:
\begin{equation}
\omega_{M_{A}M_{B}}:=\mathcal{P}_{L_{A_{n}}A_{n}B_{n}L_{B_{n}}\rightarrow
M_{A}M_{B}}^{(n+1)}(\sigma_{L_{A_{n}}A_{n}^{\prime}B_{n}^{\prime}L_{B_{n}}%
}^{(n)}),
\end{equation}
where $M_{A}$ and $M_{B}$ are held by Alice and Bob, respectively.

The goal of the protocol is for Alice and Bob to distill entanglement in the
end; i.e., the final state $\omega_{M_{A}M_{B}}$ should be close to a
maximally entangled state. For a fixed $n,\ M\in\mathbb{N},\ \varepsilon
\in\lbrack0,1]$, the original protocol is an $(n,M,\varepsilon)$ protocol if
the channel is used $n$ times as discussed above, $|M_{A}|=|M_{B}|=M$, and if
$
F(\omega_{M_{A}M_{B}},\Phi_{M_{A}M_{B}})    =\left\langle \Phi\right\vert
_{M_{A}M_{B}}\omega_{M_{A}M_{B}}\left\vert \Phi\right\rangle _{AB}
  \geq1-\varepsilon,
$
where $\Phi_{M_{A}M_{B}}$ is the maximally entangled state. A rate $R$ is said
to be achievable for entanglement distillation if for all $\varepsilon
\in(0,1]$, $\delta>0$, and sufficiently large $n$, there exists an
$(n,2^{n(R-\delta)},\varepsilon)$ protocol. The distillable entanglement of  $\mathcal{N}$, denoted as $Q%
(\mathcal{N})$, is equal to the supremum of all achievable rates.

The recent work \cite{DBW17} defined the 
max-Rains information $R_{\max}(\mathcal{N})$\ of a bipartite
quantum channel $\mathcal{N}_{A^{\prime}B^{\prime}\rightarrow AB}$\ as
follows: 
\begin{equation}
R_{\max}(\mathcal{N}):=\log\Gamma%
(\mathcal{N}),%\label{eq:bi-max-Rains-info}%
\end{equation}
where $\Gamma(\mathcal{N})$ is the solution to the following
semi-definite program:
\begin{align}
\inf  \{ \left\Vert \operatorname{Tr}_{AB}\{V_{S_{A}ABS_{B}}%
+Y_{S_{A}ABS_{B}}\}\right\Vert _{\infty}:
 V,Y\geq0,
  \ T_{BS_{B}}(V_{S_{A}ABS_{B}}-Y_{S_{A}ABS_{B}})\geq J_{S_{A}ABS_{B}%
}^{\mathcal{N}}\},\label{eq:bi-rains-channel-sdp}%
\end{align}
such that $S_{A}\simeq A^{\prime}$, and $S_{B}\simeq B^{\prime}$. One of the
main results of \cite{DBW17} is the following bound $Q%
(\mathcal{N})\leq R_{\max}(\mathcal{N})$, establishing the
max-Rains information as a fundamental limitation on the
distillable entanglement of any bipartite channel.

One of the key properties of the max-Rains information is that
it does not increase under amortization; i.e., the following inequality is
satisfied. Let $\rho_{L_{A}A^{\prime}B^{\prime}L_{B}}$ be a state, and let
$\mathcal{N}_{A^{\prime}B^{\prime}\rightarrow AB}$ be a bipartite channel.
Then
\begin{equation}
R_{\max}(L_{A}A;BL_{B})_{\omega}\leq R_{\max}(L_{A}A^{\prime};B^{\prime}%
L_{B})_{\rho}+R_{\max}(\mathcal{N}),\label{eq:amort-Rains}%
\end{equation}
where $\omega_{L_{A}ABL_{B}}=\mathcal{N}_{A^{\prime}B^{\prime}\rightarrow
AB}(\rho_{L_{A}A^{\prime}B^{\prime}L_{B}})$ and the max-Rains relative entropy
of a state $\sigma_{CD}$ is
\begin{equation}
R_{\max}(C;D)_{\sigma}:=\inf\{\lambda:\sigma_{CD}\leq2^{\lambda}\omega
_{CD},\ \omega_{CD}\geq0,\ \left\Vert T_{D}(\omega_{CD})\right\Vert _{1}%
\leq1\}.
\end{equation}
The amortization inequality above is stronger than additivity, and it is one
of the main technical tools needed for establishing the key inequality
$Q(\mathcal{N})\leq R_{\max}(\mathcal{N})$.

\subsection{Exact distillable entanglement of bipartite channels}

Another interesting question, dual to the exact entanglement cost question
proposed above, is the exact distillable entanglement of a bipartite channel.
The setting for this problem is the same as that outlined in the previous section, but
we demand that the error $\varepsilon$ is exactly equal to zero. We again consider the
free operations to be completely PPT-preserving operations, so that a solution to this
problem will give bounds for the exact distillable entanglement with LOCC.

To start out, we should recall developments for bipartite states. The most
significant progress on the exact distillable entanglement of a bipartite
state $\rho_{AB}$ has been made recently in \cite{PhysRevA.95.062322}. To begin with, let
us define the one-shot exact distillable entanglement of a bipartite state
$\rho_{AB}$ as
\begin{equation}
E_{\text{PPT}}^{(1,d)}(\rho_{AB}):=\sup\{\log_{2}d:\mathcal{P}_{AB\rightarrow
\hat{A}\hat{B}}(\rho_{AB})=\Phi_{\hat{A}\hat{B}}^{d}\},
\end{equation}
where $\mathcal{P}_{AB\rightarrow\hat{A}\hat{B}}$ is a completely PPT-preserving
operation. In \cite{PhysRevA.95.062322}, it was shown that $E_{\text{PPT}}^{(1,d)}(\rho_{AB})$ is given
by the following optimization:%
\begin{align}
E_{\text{PPT}}^{(1,d)}(\rho_{AB})  =-\log W_{0}(\rho_{AB}),\qquad
W_{0}(\rho_{AB}):=\inf\{\left\Vert T_{B}(R_{AB})\right\Vert _{\infty}%
:P_{AB}\leq R_{AB}\leq I_{AB}\},\label{eq:W_0-exact-dist-ent}%
\end{align}
with $P_{AB}$ the projection onto the support of the state $\rho_{AB}$. The
exact entanglement cost of a bipartite state $\rho_{AB}$ is then defined as
the regularization of the above:
\begin{equation}
E_{\text{PPT}}^{d}(\rho_{AB}):=\lim
_{n\rightarrow\infty}\frac{1}{n}E_{\text{PPT}}^{(1,d)}(\rho_{AB}^{\otimes n}%
).
\end{equation}
By relaxing one of the constraints for $W_{0}$ above, we get the following
quantity \cite{PhysRevA.95.062322},  called the min-Rains relative entropy:
\begin{equation}
M(\rho_{AB}):=\inf\{\left\Vert T_{B}(R_{AB})\right\Vert _{\infty}:P_{AB}\leq
R_{AB}\},
\end{equation}
and then it follows that%
\begin{equation}
E_{\text{PPT}}^{(1,d)}(\rho_{AB})\leq E_{M}(\rho_{AB}):=-\log_{2}M(\rho
_{AB}).\label{eq:min-Rains-one-shot-bnd}%
\end{equation}
However, a significant property of $E_{M}(\rho_{AB})$ is that it is
additive \cite{PhysRevA.95.062322}:
\begin{equation}
E_{M}(\rho_{AB}\otimes\sigma_{A^{\prime}B^{\prime}})=E_{M}%
(\rho_{AB})+E_{M}(\sigma_{A^{\prime}B^{\prime}}).
\end{equation}
By exploiting this
property, the following single-letter, efficiently computable
upper bound on the exact distillable entanglement  follows \cite{PhysRevA.95.062322}:%
\begin{equation}
E_{\text{PPT}}^{d}(\rho_{AB})\leq E_{M}(\rho_{AB}%
).\label{eq:min-Rains-bnd-ent-dist}%
\end{equation}

Some key questions for this task are as follows: Is 
the inequality in \eqref{eq:min-Rains-bnd-ent-dist} tight?\ This would
involve showing that one of the constraints in \eqref{eq:W_0-exact-dist-ent}
becomes negligible in the asymptotic limit of many copies of $\rho_{AB}$. If
it is true, it would be a strong counterpart to the finding in
\eqref{eq:kappa-cost}. We can also analyze the exact distillable entanglement of a
point-to-point quantum channel $\mathcal{N}_{A\rightarrow B}$, and in light of the result in
\eqref{eq:min-Rains-one-shot-bnd}, it is natural to wonder whether
\begin{equation}
E_{\text{PPT}}^{(1,d)}(\mathcal{N}_{A\rightarrow B})\leq E_{M}(\mathcal{N}%
_{A\rightarrow B}),
\end{equation}
where the one-shot distillable entanglement of a channel is given by%
\begin{equation}
E_{\text{PPT}}^{(1,d)}(\mathcal{N}_{A\rightarrow B}):=\sup\{\log
_{2}d:\mathcal{P}_{A^{\prime}BB^{\prime}\rightarrow\hat{A}\hat{B}}%
(\mathcal{N}_{A\rightarrow B}(\rho_{A^{\prime}AB^{\prime}}))=\Phi_{\hat{A}%
\hat{B}}^{d}\},
\end{equation}
with $\rho_{A^{\prime}AB^{\prime}}$ a PPT\ state and $\mathcal{P}_{A^{\prime
}BB^{\prime}\rightarrow\hat{A}\hat{B}}$ a completely PPT-preserving channel, and
the min-Rains information of a channel $\mathcal{N}_{A\rightarrow B}$ is
defined as the optimized min-Rains relative entropy:
\begin{equation}
E_{M}(\mathcal{N}_{A\rightarrow B}):=\sup_{\psi_{RA}}E_{M}(\mathcal{N}%
_{A\rightarrow B}(\psi_{RA})),
\end{equation}
where the optimization is with respect to pure states $\psi_{RA}$ with system
$R$ isomorphic to  system $A$.
From here, a natural next question is to determine bounds on the exact distillable entanglement of a
bipartite channel.

\section{Entanglement measures for bipartite channels}

Here we develop entanglement measures for bipartite channels, including logarithmic negativity, $\kappa$ entanglement, and generalized Rains information. We begin with some background and then develop the aforementioned measures.

\subsection{Entropies and information}

The quantum entropy of a density operator $\rho_A$ is defined as \cite{Neu32}
\begin{equation}
S(A)_\rho:= S(\rho_A)= -\Tr[\rho_A\log_2\rho_A].
\end{equation}
The conditional quantum entropy $S(A\vert B)_\rho$ of a density operator $\rho_{AB}$ of a composite system $AB$ is defined as
\begin{equation}
S(A\vert B)_\rho \coloneqq S(AB)_\rho-S(B)_\rho.
\end{equation}
The coherent information $I(A\> B)_{\rho}$ of a density operator $\rho_{AB}$  of a composite system $AB$ is defined as \cite{SN96}
\begin{equation}\label{eq:coh-info}
I(A\rangle B)_{\rho} \coloneqq - S(A\vert B)_\rho = S(B)_{\rho}-S(AB)_{\rho}.
\end{equation}
The quantum relative entropy of two quantum states is a measure of their distinguishability. For $\rho\in\mc{D}(\mc{H})$ and $\sigma\in\mc{B}_+(\mc{H})$, it is defined as~\cite{Ume62} 
\begin{equation}
D(\rho\V \sigma):= \left\{ 
\begin{tabular}{c c}
$\Tr\{\rho[\log_2\rho-\log_2\sigma]\}$, & $\supp(\rho)\subseteq\supp(\sigma)$\\
$+\infty$, &  otherwise.
\end{tabular} 
\right.
\end{equation}
The quantum relative entropy is non-increasing under the action of positive trace-preserving maps \cite{MR15}, which is the statement that $D(\rho\V\sigma)\geq D(\mc{M}(\rho)\V\mc{M}{(\sigma)})$ for any two density operators $\rho$ and $\sigma$ and a positive trace-preserving map $\mc{M}$ (this inequality applies to quantum channels as well \cite{Lin75}, since every completely positive map is also a positive map by definition).

\subsection{Generalized divergence and generalized relative entropies}

A quantity is called a generalized divergence \cite{PV10,SW12} if it satisfies the following monotonicity (data-processing) inequality for all density operators $\rho$ and $\sigma$ and quantum channels $\mc{N}$:
\begin{equation}\label{eq:gen-div-mono}
\mathbf{D}(\rho\Vert \sigma)\geq \mathbf{D}(\mathcal{N}(\rho)\Vert \mc{N}(\sigma)).
\end{equation}
As a direct consequence of the above inequality, any generalized divergence satisfies the following two properties for an isometry $U$ and a state~$\tau$ \cite{WWY14}:
\begin{align}
\mathbf{D}(\rho\Vert \sigma) & = \mathbf{D}(U\rho U^\dag\Vert U \sigma U^\dag),\label{eq:gen-div-unitary}\\
\mathbf{D}(\rho\Vert \sigma) & = \mathbf{D}(\rho \otimes \tau \Vert \sigma \otimes \tau).\label{eq:gen-div-prod}
\end{align}

The sandwiched R\'enyi relative entropy  \cite{MDSFT13,WWY14} is denoted as $\wt{D}_\alpha(\rho\V\sigma)$  and defined for
$\rho\in\mc{D}(\mc{H})$, $\sigma\in\mc{B}_+(\mc{H})$, and  $\forall \alpha\in (0,1)\cup(1,\infty)$ as
\begin{equation}\label{eq:def_sre}
\wt{D}_\alpha(\rho\V \sigma):= \frac{1}{\alpha-1}\log_2 \Tr\left\{\left(\sigma^{\frac{1-\alpha}{2\alpha}}\rho\sigma^{\frac{1-\alpha}{2\alpha}}\right)^\alpha \right\} ,
\end{equation}
but it is set to $+\infty$ for $\alpha\in(1,\infty)$ if $\supp(\rho)\nsubseteq \supp(\sigma)$.
%Let $\widetilde{Q}_\alpha(\rho\V\sigma):= \Tr\left\{\left(\sigma^{\frac{1-\alpha}{2\alpha}}\rho\sigma^{\frac{1-\alpha}{2\alpha}}\right)^\alpha \right\}$ for short hand notation. %when $\supp(\rho)\subseteq\supp(\sigma)$.
The sandwiched R\'enyi relative entropy obeys the following ``monotonicity in $\alpha$'' inequality \cite{MDSFT13}: for  $\alpha,\beta\in(0,1)\cup(1,\infty)$,
\begin{equation}\label{eq:mono_sre}
\wt{D}_\alpha(\rho\V\sigma)\leq \wt{D}_\beta(\rho\V\sigma) \quad \text{ if }  \quad \alpha\leq \beta.
\end{equation}
The following lemma states that the sandwiched R\'enyi relative entropy $\wt{D}_\alpha(\rho\V\sigma)$ is a particular generalized divergence for certain values of $\alpha$. 
\begin{lemma}[\cite{FL13}]%Bei13
\label{lem:DP-sandRenyi}
Let $\mc{N}:\mc{B}_+(\mc{H}_A)\to \mc{B}_+(\mc{H}_B)$ be a quantum channel   and let $\rho_A\in\mc{D}(\mc{H}_A)$ and $\sigma_A\in \mc{B}_+(\mc{H}_A)$. Then, for all $ \alpha\in \[1/2,1\)\cup (1,\infty)$
\begin{equation}
\wt{D}_\alpha(\rho\V\sigma)\geq \wt{D}_\alpha(\mc{N}(\rho)\V\mc{N}(\sigma)).
\end{equation}
\end{lemma}

See \cite{Wilde2018a} for an alternative proof of Lemma~\ref{lem:DP-sandRenyi}.

In the limit $\alpha\to 1$, the sandwiched R\'enyi relative entropy $\wt{D}_\alpha(\rho\V\sigma)$ converges to the quantum relative entropy \cite{MDSFT13,WWY14}:
\begin{equation}\label{eq:mono_renyi}
\lim_{\alpha\to 1}\wt{D}_\alpha(\rho\V\sigma):= D_1(\rho\V\sigma)=D(\rho\V\sigma).
\end{equation}
In the limit $\alpha\to \infty$, the sandwiched R\'enyi relative entropy $\wt{D}_\alpha(\rho\V\sigma)$ converges to the max-relative entropy \cite{MDSFT13}, which is defined as \cite{D09,Dat09}
\begin{equation}\label{eq:max-rel}
D_{\max}(\rho\V\sigma)=\inf\{\lambda:\ \rho \leq 2^\lambda\sigma\},
\end{equation}
and if $\supp(\rho)\nsubseteq\supp(\sigma)$ then $D_{\max}(\rho\V\sigma)=\infty$.

\subsection{Entanglement measures for bipartite states}
 
Let $\Ent(A;B)_\rho$ denote an entanglement measure \cite{HHHH09} that is evaluated for a bipartite state~$\rho_{AB}$. 
The basic property of an entanglement measure is that it should be an LOCC monotone \cite{HHHH09}, i.e.,  non-increasing under the action of an LOCC channel. 
Given such an entanglement measure, one can define the entanglement $\Ent(\mc{M})$ of a channel $\mc{M}_{A\to B}$ in terms of it by optimizing over all pure, bipartite states that can be input to the channel:
\begin{equation}\label{eq:ent-mes-channel}
\Ent(\mc{M})=\sup_{\psi_{LA}} \Ent(L;B)_\omega,
\end{equation}
where $\omega_{LB}=\mc{M}_{A\to B}(\psi_{LA})$. Due to the properties of an entanglement measure and the well known Schmidt decomposition theorem, it suffices to optimize over pure states $\psi_{LA}$ such that $L\simeq A$ (i.e., one does not achieve a higher value of
$\Ent(\mc{M})$
 by optimizing over mixed states with unbounded reference system $L$). In an information-theoretic setting, the entanglement $\Ent(\mc{M})$ of a channel~$\mc{M}$ characterizes the amount of entanglement that a sender $A$ and receiver $B$ can generate by using the channel if they do not share entanglement prior to its use.

Alternatively, one can consider the amortized entanglement $\Ent_A(\mc{M})$ of a channel $\mc{M}_{A\to B}$ as the following optimization~\cite{KW17} (see also \cite{LHL03,BHLS03,CM17,DDMW17,RKB+17}):
\begin{equation}\label{eq:ent-arm}
\Ent_A(\mc{M})
 \coloneqq \sup_{\rho_{L_AAL_B}} \left[\Ent(L_A;BL_B)_{\tau}-\Ent(L_AA;L_B)_{\rho}\right],
\end{equation}
where $\tau_{L_ABL_B}=\mc{M}_{A\to B}(\rho_{L_AAL_B})$ and $\rho_{L_AAL_B}$ is a state. The supremum is with respect to all states $\rho_{L_AAL_B}$ and the systems $L_A,L_B$ are finite-dimensional but could be arbitrarily large. Thus, in general, $\Ent_A(\mc{M})$ need not be computable. The amortized entanglement quantifies the net amount of entanglement that can be generated by using the channel $\mc{M}_{A\to B}$, if the sender and the receiver are allowed to begin with some initial entanglement in the form of the state $\rho_{L_AAL_B}$. That is, $\Ent(L_AA;L_B)_\rho$ quantifies the entanglement of the initial state $\rho_{L_AAL_B}$, and $\Ent(L_A;BL_B)_{\tau}$ quantifies the entanglement of the final state produced after the action of the channel. 

The Rains relative entropy of a state $\rho_{AB}$ is defined as \cite{Rai01,AdMVW02}
\begin{equation}\label{eq:rains-inf-state}
R(A;B)_\rho\coloneqq \min_{\sigma_{AB}\in \PPT'(A:B)} D (\rho_{AB}\Vert \sigma_{AB}),
\end{equation}
and it is monotone non-increasing under the action of a completely PPT-preserving quantum channel $\mc{P}_{A'B'\to AB}$, i.e.,
\begin{equation}
R(A';B')_\rho\geq R(A;B)_\omega,
\end{equation}
where $\omega_{AB}=\mc{P}_{A'B'\to AB}(\rho_{A'B'})$. The sandwiched Rains relative entropy of a state $\rho_{AB}$ is defined as follows \cite{TWW17}:  
\begin{equation}\label{eq:alpha-rains-inf-state}
\widetilde{R}_{\alpha}(A;B)_\rho\coloneqq \min_{\sigma_{AB}\in \PPT'(A:B)} \widetilde{D}_{\alpha} (\rho_{AB}\Vert \sigma_{AB}).
\end{equation}
The max-Rains relative entropy of a state $\rho_{AB}$ is defined as \cite{WD16b}
\begin{equation}
R_{\max}(A;B)_\rho\coloneqq \min_{\sigma_{AB}\in \PPT'(A:B)} D_{\max} (\rho_{AB}\Vert \sigma_{AB}).
\end{equation}
The max-Rains information of a quantum channel $\mc{M}_{A\to B}$ is defined as \cite{WFD17}
\begin{equation}
R_{\max}(\mc{M})\coloneqq \max_{\phi_{SA}}R_{\max} (S;B)_\omega,
\label{eq:max-Rains-channel}
\end{equation}
where $\omega_{SB}=\mc{M}_{A\to B}(\phi_{SA})$ and $\phi_{SA}$ is a pure state, with $\dim(\mc{H}_S)=\dim(\mc{H}_A)$. The amortized max-Rains information of a channel $\mc{M}_{A\to B}$, denoted as $R_{\max,A}(\mc{M})$, is defined by replacing $\Ent$ in \eqref{eq:ent-arm} with the max-Rains relative entropy $R_{\max}$ \cite{BW17}. It was shown in \cite{BW17} that amortization does not enhance the max-Rains information of an arbitrary point-to-point channel, i.e.,
\begin{equation}
R_{\max,A}(\mc{M})=R_{\max}(\mc{M}).
\end{equation}  

Recently, in \cite[Eq.~(8)]{WD16pra} (see also \cite{WFD17}), the max-Rains relative entropy of a state $\rho_{AB}$ was expressed as 
\begin{equation}\label{eq:rains-w}
R_{\max}(A;B)_{\rho}=\log_2 W(A;B)_{\rho}, 
\end{equation}
where $W(A;B)_{\rho}$ is the solution to the following semi-definite program:
\begin{align}
\textnormal{minimize}\ &\ \Tr\{C_{AB}+D_{AB}\}\nonumber\\
\textnormal{subject to}\ &\ C_{AB}, D_{AB}\geq 0,\nonumber\\
   &\ \T_{B} (C_{AB}-D_{AB})\geq \rho_{AB}. \label{eq:rains-state-sdp}
\end{align}
Similarly, in \cite[Eq.~(21)]{WFD17}, the max-Rains information of a quantum channel $\mc{M}_{A\to B}$ was expressed as 
\begin{equation}\label{eq:rains-omega}
R_{\max}(\mc{M})=\log \Gamma (\mc{M}),
\end{equation}
where $\Gamma(\mc{M})$ is the solution to the following semi-definite program:
\begin{align}
\textnormal{minimize}\ &\ \norm{\Tr_B\{V_{SB}+Y_{SB}\}}_{\infty}\nonumber\\
\textnormal{subject to}\ &\ Y_{SB}, V_{SB}\geq 0,\nonumber\\
& \ \T_B(V_{SB}-Y_{SB})\geq J^\mc{M}_{SB}.\label{eq:rains-channel-sdp}
\end{align}

The sandwiched relative entropy of entanglement of a bipartite state $\rho_{AB}$ is defined as \cite{WTB16} 
\begin{equation}\label{eq:rel-ent-state}
\widetilde{E}_{\alpha}(A;B)_{\rho}\coloneqq\min_{\sigma_{AB}\in\SEP(A:B)}\widetilde{D}_{\alpha}(\rho_{AB}\Vert\sigma_{AB}).
\end{equation} 
In the limit $\alpha\to 1$, $\widetilde{E}_{\alpha}(A;B)_{\rho}$ converges to the relative entropy of entanglement \cite{VP98}, i.e.,
\begin{align}\label{eq:rel-ent-state-1}
\lim_{\alpha\to 1}\widetilde{E}_{\alpha}(A;B)_{\rho} &=E(A;B)_{\rho}\\
& \coloneqq \min_{\sigma_{AB}\in\SEP(A:B)}D(\rho_{AB}\Vert\sigma_{AB}).
\end{align} The max-relative entropy of entanglement \cite{D09,Dat09} is defined for a bipartite state $\rho_{AB}$ as
\begin{equation}\label{eq:Emax}
E_{\max}(A;B)_{\rho}\coloneqq \min_{\sigma_{AB}\in\SEP(A:B)}D_{\max}(\rho_{AB}\Vert\sigma_{AB}).
\end{equation}  
The max-relative entropy of entanglement $E_{\max}(\mc{M})$ of a channel $\mc{M}_{A\to B}$ is defined as in \eqref{eq:ent-mes-channel}, by replacing $\Ent$ with $E_{\max}$ \cite{CM17}. It was shown in \cite{CM17} that amortization does not increase max-relative entropy of entanglement of a channel $\mc{M}_{A\to B}$, i.e.,
\begin{equation}
E_{\max,A}(\mc{M})=E_{\max}(\mc{M}).
\end{equation}

\subsection{Negativity of a bipartite state}

Given a bipartite state, its logarithmic negativity is defined as \cite{Vidal2002,Plenio2005b}
\begin{equation}
E_{N}(\rho_{AB}):=\log\left\Vert T_{B}(\rho_{AB})\right\Vert _{1}.
\end{equation}
The idea of this quantity is to quantify the deviation of a bipartite state
from being PPT. If it is indeed PPT, then $E_{N}(\rho_{AB})=0$. If not, then
$E_{N}(\rho_{AB})>0$.

By utilizing Holder duality, it is possible to write the above as a
semi-definite program:%
\begin{equation}
E_{N}(\rho_{AB})=\log\sup_{R_{AB}}\left\{  \operatorname{Tr}[R_{AB}\rho
_{AB}]:-I_{AB}\leq T_{B}(R_{AB})\leq I_{AB}\right\}  ,
\end{equation}
where the optimization is with respect to Hermitian $R_{AB}$. By utilizing
semi-definite programming duality, we can also write $E_{N}(\rho_{AB})$ in
terms of its dual semi-definite program as%
\begin{equation}
E_{N}(\rho_{AB})=\log\inf_{K_{AB},L_{AB}\geq0}\left\{  \operatorname{Tr}%
[K_{AB}+L_{AB}]:T_{B}(K_{AB}-L_{AB})=\rho_{AB}\right\}  .
\end{equation}

The max-Rains relative entropy of a bipartite state is defined as follows \cite{WD16pra}:%
\begin{equation}
R_{\max}(\rho_{AB}):=\inf_{\sigma_{AB}\geq0,E_{N}(\sigma_{AB})\leq
0}D_{\max}(\rho_{AB}\Vert\sigma_{AB}). \label{eq:Rains-rel-ent}%
\end{equation}
It can be written as the following semi-definite program:%
\begin{equation}
R_{\max}(\rho_{AB})=\log\sup_{R_{AB}\geq0}\left\{  \operatorname{Tr}%
[R_{AB}\rho_{AB}]:-I_{AB}\leq T_{B}(R_{AB})\leq I_{AB}\right\}  ,
\end{equation}
with the dual%
\begin{equation}
R_{\max}(\rho_{AB})=\log\inf_{K_{AB},L_{AB}\geq0}\left\{  \operatorname{Tr}%
[K_{AB}+L_{AB}]:T_{B}(K_{AB}-L_{AB})\geq\rho_{AB}\right\}  .
\end{equation}

It is clear that%
\begin{equation}
R_{\max}(\rho_{AB})\leq E_{N}(\rho_{AB}),
\end{equation}
since the primal for $R_{\max}(\rho_{AB})$ is obtained from the primal for
$E_{N}(\rho_{AB})$ by restricting the optimization to $R_{AB}\geq0$.
Alternatively, the dual of $R_{\max}(\rho_{AB})$ is obtained from the dual of
$E_{N}(\rho_{AB})$ by relaxing the equality constraint $\rho_{AB}=T_{B}%
(K_{AB}-L_{AB})$.

Finally, note that we can define Rains relative entropy of a bipartite state
much more generally in terms of a generalized divergence $\mathbf{D}$ as%
\begin{equation}
\mathbf{R}(\rho_{AB}):=\inf_{\sigma_{AB}\geq0,E_{N}(\sigma_{AB})\leq
0}\mathbf{D}(\rho_{AB}\Vert\sigma_{AB})
\end{equation}

\subsection{Negativity of a bipartite channel}

Let us define the logarithmic negativity of a bipartite channel $\mathcal{N}%
_{A^{\prime}B^{\prime}\rightarrow AB}$\ as%
\begin{equation}
E_{N}(\mathcal{N}):=\log\left\Vert T_{B}\circ\mathcal{N}_{A^{\prime
}B^{\prime}\rightarrow AB}\circ T_{B^{\prime}}\right\Vert _{\Diamond
},\label{eq:log-neg-channel}%
\end{equation}
where the diamond norm \cite{Kit97} of a bipartite linear, Hermitian-preserving map
$\mathcal{P}_{A^{\prime}B^{\prime}\rightarrow AB}$\ is given by%
\begin{equation}
\left\Vert \mathcal{P}_{A^{\prime}B^{\prime}\rightarrow AB}\right\Vert
_{\Diamond}:=\log\sup_{\psi_{S_{A}A^{\prime}B^{\prime}S_{B}}}\left\Vert
\mathcal{P}_{A^{\prime}B^{\prime}\rightarrow AB}(\psi_{S_{A}A^{\prime
}B^{\prime}S_{B}})\right\Vert _{1}.\label{eq:diamond-norm}%
\end{equation}
Thus, more generally, $E_{N}(\mathcal{N})$ can be defined in the above way if
$\mathcal{N}_{A^{\prime}B^{\prime}\rightarrow AB}$ is an arbitrary linear,
Hermitian-preserving map. Note that $E_{N}(\mathcal{N})$ reduces to the well known logarithmic negativity of a point-to-point channel \cite{HW01} when the bipartite channel is indeed a point-to-point channel.

A bipartite channel $\mathcal{N}_{A^{\prime}B^{\prime}\rightarrow AB}$ is
called completely PPT preserving (C-PPT-P) if the map $T_{B}\circ
\mathcal{N}_{A^{\prime}B^{\prime}\rightarrow AB}\circ T_{B^{\prime}}$ is
completely positive \cite{Rai99,Rai01}. Thus, the measure in
\eqref{eq:log-neg-channel}\ quantifies the deviation of a bipartite
channel from being C-PPT-P. Indeed, if $\mathcal{N}_{A^{\prime}B^{\prime
}\rightarrow AB}$ is C-PPT-P, then $E_{N}(\mathcal{N})=0$. Otherwise,
$E_{N}(\mathcal{N})>0$. 

\begin{proposition}
The logarithmic negativity of a bipartite channel $\mathcal{N}_{A^{\prime
}B^{\prime}\rightarrow AB}$ can be written as the following primal SDP:%
\begin{equation}
\sup_{\rho,R}\left\{
\begin{array}
[c]{c}%
\operatorname{Tr}[T_{BS_{B}}(J_{S_{A}ABS_{B}}^{\mathcal{N}})R_{S_{A}ABS_{B}%
}]:\rho_{S_{A}S_{B}}\geq0,\operatorname{Tr}[\rho_{S_{A}S_{B}}]\leq1,\\
-\rho_{S_{A}S_{B}}\otimes I_{AB}\leq R_{S_{A}ABS_{B}}\leq\rho_{S_{A}S_{B}%
}\otimes I_{AB}%
\end{array}
\right\}  ,
\end{equation}
where $J_{S_{A}ABS_{B}}^{\mathcal{N}}$ is the Choi operator of the channel
$\mathcal{N}_{A^{\prime}B^{\prime}\rightarrow AB}$ and the optimization is
with respect to Hermitian $R_{S_{A}ABS_{B}}$. The dual SDP\ is given by%
\begin{equation}
\inf\left\{
\begin{array}
[c]{c}%
\left\Vert \operatorname{Tr}_{AB}[V_{S_{A}ABS_{B}}+Y_{S_{A}ABS_{B}%
}]\right\Vert _{\infty}:V_{S_{A}ABS_{B}},Y_{S_{A}ABS_{B}}\geq0,\\
T_{BS_{B}}(V_{S_{A}ABS_{B}}-Y_{S_{A}ABS_{B}})=J_{S_{A}ABS_{B}}^{\mathcal{N}}%
\end{array}
\right\}  .
\end{equation}

\end{proposition}

\begin{proof}
Starting from the definition, we find that%
\begin{align}
&  \sup_{\psi}\left\Vert (T_{B}\circ\mathcal{N}_{A^{\prime}B^{\prime
}\rightarrow AB}\circ T_{B^{\prime}})(\psi_{S_{A}A^{\prime}B^{\prime}S_{B}%
})\right\Vert _{1}\\
&  =\sup_{\psi,R}\left\{  \operatorname{Tr}[(T_{B}\circ\mathcal{N}_{A^{\prime
}B^{\prime}\rightarrow AB}\circ T_{B^{\prime}})(\psi_{S_{A}A^{\prime}%
B^{\prime}S_{B}})R_{S_{A}ABS_{B}}]:\left\Vert R_{S_{A}ABS_{B}}\right\Vert
_{\infty}\leq1\right\}  \\
&  =\sup_{\rho,R}\left\{
\begin{array}
[c]{c}%
\operatorname{Tr}[(T_{B}\circ\mathcal{N}_{A^{\prime}B^{\prime}\rightarrow
AB}\circ T_{B^{\prime}})(\rho_{S_{A}S_{B}}^{1/2}\Gamma_{S_{A}A^{\prime
}B^{\prime}S_{B}}\rho_{S_{A}S_{B}}^{1/2})R_{S_{A}ABS_{B}}]:\\
\rho_{S_{A}S_{B}}\geq0,\operatorname{Tr}[\rho_{S_{A}S_{B}}]=1,\left\Vert
R_{S_{A}ABS_{B}}\right\Vert _{\infty}\leq1
\end{array}
\right\}  \\
&  =\sup_{\rho,R}\left\{
\begin{array}
[c]{c}%
\operatorname{Tr}[\rho_{S_{A}S_{B}}^{1/2}T_{BS_{B}}(J_{S_{A}ABS_{B}%
}^{\mathcal{N}})\rho_{S_{A}S_{B}}^{1/2}R_{S_{A}ABS_{B}}]:\rho_{S_{A}S_{B}}%
\geq0,\operatorname{Tr}[\rho_{S_{A}S_{B}}]=1,\\
-I_{S_{A}ABS_{B}}\leq R_{S_{A}ABS_{B}}\leq I_{S_{A}ABS_{B}}%
\end{array}
\right\}  \\
&  =\sup_{\rho,R}\left\{
\begin{array}
[c]{c}%
\operatorname{Tr}[T_{BS_{B}}(J_{S_{A}ABS_{B}}^{\mathcal{N}})\rho_{S_{A}S_{B}%
}^{1/2}R_{S_{A}ABS_{B}}\rho_{S_{A}S_{B}}^{1/2}]:\rho_{S_{A}S_{B}}%
\geq0,\operatorname{Tr}[\rho_{S_{A}S_{B}}]=1,\\
-I_{S_{A}ABS_{B}}\leq R_{S_{A}ABS_{B}}\leq I_{S_{A}ABS_{B}}%
\end{array}
\right\}  \\
&  =\sup_{\rho,R}\left\{
\begin{array}
[c]{c}%
\operatorname{Tr}[T_{BS_{B}}(J_{S_{A}ABS_{B}}^{\mathcal{N}})R_{S_{A}ABS_{B}%
}]:\rho_{S_{A}S_{B}}\geq0,\operatorname{Tr}[\rho_{S_{A}S_{B}}]=1,\\
-\rho_{S_{A}S_{B}}\otimes I_{AB}\leq R_{S_{A}ABS_{B}}\leq\rho_{S_{A}S_{B}%
}\otimes I_{AB}%
\end{array}
\right\}
\end{align}
Thus, it is clearly an SDP. By employing standard techniques, we find that the
dual is given as stated in the proposition.
\end{proof}

\begin{proposition}
[Faithfulness]The logarithmic negativity of a bipartite channel $\mathcal{N}%
_{A^{\prime}B^{\prime}\rightarrow AB}$\ obeys the following faithfulness
condition:%
\begin{equation}
E_{N}(\mathcal{N})\geq0\text{ and }E_{N}(\mathcal{N})=0\text{ if and only if
}\mathcal{N}\in\text{\textrm{C-PPT-P}.}%
\end{equation}

\end{proposition}

\begin{proof}
The first inequality is equivalent to the following one:%
\begin{equation}
\left\Vert T_{B}\circ\mathcal{N}_{A^{\prime}B^{\prime}\rightarrow AB}\circ
T_{B^{\prime}}\right\Vert _{\Diamond}\geq1.
\end{equation}
Pick $\psi_{S_{A}A^{\prime}B^{\prime}S_{B}}$\ in \eqref{eq:diamond-norm} to be
$\Phi_{S_{A}A^{\prime}}\otimes\Phi_{B^{\prime}S_{B}}$. Then%
\begin{align}
(T_{B}\circ\mathcal{N}_{A^{\prime}B^{\prime}\rightarrow AB}\circ T_{B^{\prime
}})(\Phi_{S_{A}A^{\prime}}\otimes\Phi_{B^{\prime}S_{B}})  & =(T_{B}%
\circ\mathcal{N}_{A^{\prime}B^{\prime}\rightarrow AB}\circ T_{S_{B}}%
)(\Phi_{S_{A}A^{\prime}}\otimes\Phi_{B^{\prime}S_{B}})\\
& =(T_{BS_{B}}\circ\mathcal{N}_{A^{\prime}B^{\prime}\rightarrow AB}%
)(\Phi_{S_{A}A^{\prime}}\otimes\Phi_{B^{\prime}S_{B}})\\
& =T_{BS_{B}}(\Phi_{S_{A}ABS_{B}}^{\mathcal{N}}),
\end{align}
where $\Phi_{S_{A}ABS_{B}}^{\mathcal{N}}$ denotes the Choi state of the
channel $\mathcal{N}_{A^{\prime}B^{\prime}\rightarrow AB}$. Then%
\begin{equation}
\left\Vert T_{B}\circ\mathcal{N}_{A^{\prime}B^{\prime}\rightarrow AB}\circ
T_{B^{\prime}}\right\Vert _{\Diamond}\geq\left\Vert T_{BS_{B}}(\Phi
_{S_{A}ABS_{B}}^{\mathcal{N}})\right\Vert _{1}\geq1,
\end{equation}
the latter inequality following from the faithfulness of the logarithmic
negativity of states.

Now suppose that $\mathcal{N}_{A^{\prime}B^{\prime}\rightarrow AB}\in
$\textrm{C-PPT-P}. Then it follows that $T_{B}\circ\mathcal{N}_{A^{\prime
}B^{\prime}\rightarrow AB}\circ T_{B^{\prime}}$ is a quantum channel, so that%
\begin{equation}
\left\Vert T_{B}\circ\mathcal{N}_{A^{\prime}B^{\prime}\rightarrow AB}\circ
T_{B^{\prime}}\right\Vert _{\Diamond}=1
\end{equation}
and thus $E_{N}(\mathcal{N})=0$.

Now suppose that $E_{N}(\mathcal{N})=0$. Then%
\begin{equation}
\left\Vert T_{B}\circ\mathcal{N}_{A^{\prime}B^{\prime}\rightarrow AB}\circ
T_{B^{\prime}}\right\Vert _{\Diamond}=1,
\end{equation}
and thus%
\begin{equation}
\left\Vert T_{BS_{B}}(\Phi_{S_{A}ABS_{B}}^{\mathcal{N}})\right\Vert _{1}=1.
\end{equation}
From the faithfulness condition of logarithmic negativity of states, it
follows that $T_{BS_{B}}(\Phi_{S_{A}ABS_{B}}^{\mathcal{N}})\in$PPT. However,
it is known from the work \cite{Rai99,Rai01} that this condition is equivalent to
$\mathcal{N}_{A^{\prime}B^{\prime}\rightarrow AB}\in$\textrm{C-PPT-P}.
\end{proof}

\bigskip
A PPT superchannel $\Theta^{\operatorname{PPT}}$ is a physical transformation of a bipartite quantum channel. That is, the superchannel realizes the following transformation of a channel
$\mathcal{M}_{\hat{A}'\hat{B}'\rightarrow\hat{A}\hat{B}}$
to a channel $\mathcal{N}_{A\rightarrow B}$
in terms of 
completely-PPT-preserving channels $\mathcal{P}_{A'B'\rightarrow
\hat{A}'\hat{B}'A_{M}B_{M}}^{\text{pre}}$ and $\mathcal{P}_{A_{M}\hat{A}\hat{B}B_{M} \to AB}^{\text{post}}$:
\begin{equation}
\mathcal{N}_{A'B'\rightarrow AB}=\Theta^{\operatorname{PPT}}(\mathcal{M}_{\hat
{A}\rightarrow\hat{B}})\coloneqq\mathcal{P}_{A_{M}\hat{A}\hat{B}B_{M} \to AB}^{\text{post}}%
\circ\mathcal{M}_{\hat{A}'\hat{B}'\rightarrow\hat{A}\hat{B}}\circ\mathcal{P}_{A'B'\rightarrow
\hat{A}'\hat{B}'A_{M}B_{M}}^{\text{pre}}.
\label{eq:superchannel-action}
\end{equation}

\begin{theorem}
[Monotonicity]
Let $\mathcal{M}_{\hat{A}'\hat{B}'\rightarrow\hat{A}\hat{B}
}$ be a bipartite quantum channel and  $\Theta^{\operatorname{PPT}}$ a completely-PPT-preserving superchannel of the form in \eqref{eq:superchannel-action}.
The channel measure $E_N$ is monotone
under the action of the superchannel $\Theta^{\operatorname{PPT}}$, in the sense that%
\begin{equation}
E_N(\mathcal{M}_{\hat{A}'\hat{B}'\rightarrow\hat{A}\hat{B}
}) 
\geq 
E_N(\Theta^{\operatorname{PPT}}(\mathcal{M}_{\hat{A}'\hat{B}'\rightarrow\hat{A}\hat{B}
})).
\end{equation}

\end{theorem}

\begin{proof}
Follows from the definition of $E_N$, structure of PPT superchannels, and properties of the diamond norm.
\end{proof}

\subsection{Generalized Rains information of a bipartite channel}

Recall that the max-divergence of completely positive maps $\mathcal{E}%
_{C\rightarrow D}$ and $\mathcal{F}_{C\rightarrow D}$ is defined as
\cite{CMW14}%
\begin{equation}
D_{\max}(\mathcal{E}\Vert\mathcal{F})=\sup_{\psi_{RC}}D_{\max}(\mathcal{E}%
_{C\rightarrow D}(\psi_{RC})\Vert\mathcal{F}_{C\rightarrow D}(\psi_{RC})),
\end{equation}
where the optimization is with respect to all pure bipartite states with
reference system $R$ isomorphic to the channel input system $C$. We then
define the max-Rains information of a bipartite channel as a
generalization of the state measure in \eqref{eq:Rains-rel-ent}:

\begin{definition}
The max-Rains information of a bipartite channel $\mathcal{N}_{A^{\prime
}B^{\prime}\rightarrow AB}$ is defined as%
\begin{equation}
R_{\max}(\mathcal{N}):=\inf_{\mathcal{M}:E_{N}(\mathcal{M})\leq0}D_{\max
}(\mathcal{N}\Vert\mathcal{M}),
\end{equation}
where the minimization is with respect to all completely positive bipartite
maps $\mathcal{M}_{A^{\prime}B^{\prime}\rightarrow AB}$.
The generalized Rains information of a bipartite
channel $\mathcal{N}_{A^{\prime
}B^{\prime}\rightarrow AB}$ is defined as%
\begin{equation}
\mathbf{R}(\mathcal{N}):=\inf_{\mathcal{M}:E_{N}(\mathcal{M})\leq
0}\mathbf{D}(\mathcal{N}\Vert\mathcal{M}),
\end{equation}
by utilizing a generalized channel divergence $\mathbf{D}$. \end{definition}

\begin{theorem}
[Monotonicity]
Let $\mathcal{M}_{\hat{A}'\hat{B}'\rightarrow\hat{A}\hat{B}
}$ be a bipartite quantum channel and  $\Theta^{\operatorname{PPT}}$ a completely-PPT-preserving superchannel of the form in \eqref{eq:superchannel-action}.
The channel measure $\mathbf{R}(\mathcal{N})$ is monotone
under the action of the superchannel $\Theta^{\operatorname{PPT}}$, in the sense that%
\begin{equation}
\mathbf{R}(\mathcal{N})(\mathcal{M}_{\hat{A}'\hat{B}'\rightarrow\hat{A}\hat{B}
}) 
\geq 
\mathbf{R}(\mathcal{N})(\Theta^{\operatorname{PPT}}(\mathcal{M}_{\hat{A}'\hat{B}'\rightarrow\hat{A}\hat{B}
})).
\end{equation}

\end{theorem}

\begin{proof} The proof is similar to Theorem~10 of \cite{WWS19channels}.
Follows from the definition of $\mathbf{R}(\mathcal{N})$, its data processing property, and the structure of PPT superchannels.
\end{proof}

\begin{proposition}
The max-Rains information of the bipartite channel
$\mathcal{N}_{A^{\prime}B^{\prime}\rightarrow AB}$ can be written as%
\begin{equation}
\log\inf\left\{  \left\Vert T_{B}\circ\mathcal{M}_{A^{\prime}B^{\prime
}\rightarrow AB}\circ T_{B^{\prime}}\right\Vert _{\Diamond}:J^{\mathcal{N}%
}\leq J^{\mathcal{M}}\right\}  ,
\end{equation}
where the minimization is with respect to all completely positive bipartite
maps $\mathcal{M}_{A^{\prime}B^{\prime}\rightarrow AB}$.
\end{proposition}

\begin{proof}
This follows because%
\begin{align}
&  \inf_{\mathcal{M}_{A^{\prime}B^{\prime}\rightarrow AB}:E_{N}(\mathcal{M}%
)\leq0}D_{\max}(\mathcal{N}\Vert\mathcal{M})\\
&  =\log\inf_{\lambda,\mathcal{M}_{A^{\prime}B^{\prime}\rightarrow AB}%
:E_{N}(\mathcal{M})\leq0}\left\{  \lambda:J^{\mathcal{N}}\leq\lambda
J^{\mathcal{M}}\right\}  \\
&  =\log\inf\left\{  \lambda:J^{\mathcal{N}}\leq\lambda J^{\mathcal{M}%
},\ \left\Vert T_{B}\circ\mathcal{M}_{A^{\prime}B^{\prime}\rightarrow AB}\circ
T_{B^{\prime}}\right\Vert _{\Diamond}\leq1\right\}  \\
&  =\log\inf\left\{  \lambda:J^{\mathcal{N}}\leq\lambda J^{\mathcal{M}%
},\ \left\Vert T_{B}\circ\lambda\mathcal{M}_{A^{\prime}B^{\prime}\rightarrow
AB}\circ T_{B^{\prime}}\right\Vert _{\Diamond}\leq\lambda\right\}  \\
&  =\log\inf\left\{  \lambda:J^{\mathcal{N}}\leq J^{\mathcal{M}},\ \left\Vert
T_{B}\circ\mathcal{M}_{A^{\prime}B^{\prime}\rightarrow AB}\circ T_{B^{\prime}%
}\right\Vert _{\Diamond}\leq\lambda\right\}  \\
&  =\log\inf\left\{  \left\Vert T_{B}\circ\mathcal{M}_{A^{\prime}B^{\prime
}\rightarrow AB}\circ T_{B^{\prime}}\right\Vert _{\Diamond}:J^{\mathcal{N}%
}\leq J^{\mathcal{M}}\right\}  ,
\end{align}
concluding the proof.
\end{proof}

\begin{proposition}
The max-Rains information of a bipartite channel $\mathcal{N}%
_{A^{\prime}B^{\prime}\rightarrow AB}$\ can be expressed as the following
semi-definite program:%
\begin{equation}
R_{\max}(\mathcal{N})=\log\inf\left\{
\begin{array}
[c]{c}%
\left\Vert \operatorname{Tr}_{AB}[V_{S_{A}ABS_{B}}+Y_{S_{A}ABS_{B}%
}]\right\Vert _{\infty}:V_{S_{A}ABS_{B}},Y_{S_{A}ABS_{B}}\geq0,\\
T_{BS_{B}}(V_{S_{A}ABS_{B}}-Y_{S_{A}ABS_{B}})\geq J_{S_{A}ABS_{B}%
}^{\mathcal{N}}%
\end{array}
\right\}  ,
\end{equation}
and is thus equivalent to the definition presented in \cite{DBW17}. The dual
SDP\ is given by%
\begin{equation}
R_{\max}(\mathcal{N})=\log\sup\left\{
\begin{array}
[c]{c}%
\operatorname{Tr}[J_{S_{A}ABS_{B}}^{\mathcal{N}}X_{S_{A}ABS_{B}}%
]:X_{S_{A}ABS_{B}},\rho_{S_{A}SB}\geq0,\operatorname{Tr}[\rho_{S_{A}S_{B}%
}=1],\\
-\rho_{S_{A}S_{B}}\otimes I_{AB}\leq T_{BS_{B}}(X_{S_{A}ABS_{B}})\leq
\rho_{S_{A}S_{B}}\otimes I_{AB}%
\end{array}
\right\}  ,
\end{equation}
which coincides with what was presented in \cite{DBW17}.
\end{proposition}

\begin{proof}
Consider that%
\begin{align}
&  \inf_{\mathcal{M}_{A^{\prime}B^{\prime}\rightarrow AB}:E_{N}(\mathcal{M}%
)\leq0}D_{\max}(\mathcal{N}\Vert\mathcal{M})\\
&  =\log\inf\left\{  \left\Vert T_{B}\circ\mathcal{M}_{A^{\prime}B^{\prime
}\rightarrow AB}\circ T_{B^{\prime}}\right\Vert _{\Diamond}:J^{\mathcal{N}%
}\leq J^{\mathcal{M}}\right\}  \\
&  =\log\inf\left\{
\begin{array}
[c]{c}%
\left\Vert \operatorname{Tr}_{AB}[V_{S_{A}ABS_{B}}+Y_{S_{A}ABS_{B}%
}]\right\Vert _{\infty}:J^{\mathcal{N}}\leq J^{\mathcal{M}},\ V_{S_{A}ABS_{B}%
},Y_{S_{A}ABS_{B}}\geq0,\\
T_{BS_{B}}(V_{S_{A}ABS_{B}}-Y_{S_{A}ABS_{B}})=J_{S_{A}ABS_{B}}^{\mathcal{M}}%
\end{array}
\right\}  \\
&  =\log\inf\left\{
\begin{array}
[c]{c}%
\left\Vert \operatorname{Tr}_{AB}[V_{S_{A}ABS_{B}}+Y_{S_{A}ABS_{B}%
}]\right\Vert _{\infty}:V_{S_{A}ABS_{B}},Y_{S_{A}ABS_{B}}\geq0,\\
T_{BS_{B}}(V_{S_{A}ABS_{B}}-Y_{S_{A}ABS_{B}})\geq J_{S_{A}ABS_{B}%
}^{\mathcal{N}}%
\end{array}
\right\}  ,
\end{align}
where the last equality follows from eliminating the redundant variable
$J^{\mathcal{M}}$. The dual formulation follows from standard techniques of
semi-definite programming duality.
\end{proof}

\begin{proposition}
[Reduction to states]Let $\mathcal{N}_{A^{\prime}B^{\prime}\rightarrow AB}$ be
a bipartite replacer channel, having the following action on an arbitrary
input state $\rho_{A^{\prime}B^{\prime}}$:%
\begin{equation}
\mathcal{N}_{A^{\prime}B^{\prime}\rightarrow AB}(\rho_{A^{\prime}B^{\prime}%
})=\operatorname{Tr}[\rho_{A^{\prime}B^{\prime}}]\omega_{AB},
\end{equation}
where $\omega_{AB}$ is some state. Then%
\begin{equation}
E_{N}(\mathcal{N})=E_{N}(\omega_{AB}),\ \ \ \mathbf{R}(\mathcal{N}%
)=\mathbf{R}(\omega_{AB}).
\end{equation}

\end{proposition}

\begin{proof}
For the negativity, this follows because%
\begin{align}
E_{N}(\mathcal{N}) &  =\log\sup_{\psi_{S_{A}A^{\prime}B^{\prime}S_{B}}%
}\left\Vert (T_{B}\circ\mathcal{N}_{A^{\prime}B^{\prime}\rightarrow AB}\circ
T_{B^{\prime}})(\psi_{S_{A}A^{\prime}B^{\prime}S_{B}})\right\Vert _{1}\\
&  =\log\sup_{\psi_{S_{A}A^{\prime}B^{\prime}S_{B}}}\left\Vert (T_{B}%
\circ\left(  \operatorname{Tr}_{A^{\prime}B^{\prime}}[T_{B^{\prime}}%
(\psi_{S_{A}A^{\prime}B^{\prime}S_{B}})]\omega_{AB}\right)  \right\Vert _{1}\\
&  =\log\sup_{\psi_{S_{A}A^{\prime}B^{\prime}S_{B}}}\left\Vert (T_{B}%
\circ\left(  \operatorname{Tr}_{A^{\prime}B^{\prime}}[\psi_{S_{A}A^{\prime
}B^{\prime}S_{B}}]\omega_{AB}\right)  \right\Vert _{1}\\
&  =\log\sup_{\psi_{S_{A}A^{\prime}B^{\prime}S_{B}}}\left\Vert T_{B}%
(\omega_{AB})\otimes\psi_{S_{A}S_{B}}\right\Vert _{1}\\
&  =\log\left\Vert T_{B}(\omega_{AB})\right\Vert _{1}\\
&  =E_{N}(\omega_{AB}).
\end{align}

For the other equality, denoting the maximally mixed state by $\pi$, consider
that
\begin{align}
\mathbf{R}(\mathcal{N}) &  =\inf_{\mathcal{M}:E_{N}(\mathcal{M})\leq0}%
\sup_{\psi_{S_{A}A^{\prime}B^{\prime}S_{B}}}\mathbf{D}((\operatorname{id}%
_{R}\otimes\mathcal{N})(\psi_{S_{A}A^{\prime}B^{\prime}S_{B}})\Vert
(\operatorname{id}_{R}\otimes\mathcal{M})(\psi_{S_{A}A^{\prime}B^{\prime}%
S_{B}}))\\
&  \geq\inf_{\mathcal{M}:E_{N}(\mathcal{M})\leq0}\mathbf{D}(\pi_{S_{A}S_{B}%
}\otimes\mathcal{N}(\pi_{A^{\prime}B^{\prime}})\Vert\pi_{S_{A}S_{B}}%
\otimes\mathcal{M}(\pi_{A^{\prime}B^{\prime}}))\\
&  =\inf_{\mathcal{M}:E_{N}(\mathcal{M})\leq0}\mathbf{D}(\omega_{AB}%
\Vert\mathcal{M}(\pi_{A^{\prime}B^{\prime}}))\\
&  =\inf_{\tau_{AB}:E_{N}(\tau_{AB})\leq0}\mathbf{D}(\omega_{AB}\Vert\tau
_{AB})\\
&  =\mathbf{R}(\sigma).
\end{align}
The first equality follows from the definition. The inequality follows by
choosing the input state suboptimally to be $\pi_{R}\otimes\pi_{A}$. The
second equality follows because the max-relative entropy is invariant with
respect to tensoring in the same state for both arguments. The third equality
follows because $\pi_{A^{\prime}B^{\prime}}$ is a free state in $\left\{
\tau_{AB}\geq0,E_{N}(\tau_{AB})\leq0\right\}  $ and $\mathcal{M}$ is a
completely positive map with $E_{N}(\mathcal{M})\leq0$. Since one can reach
all and only the operators in $\left\{  \tau_{AB}\geq0,E_{N}(\tau_{AB}%
)\leq0\right\}  $, the equality follows. Then the last equality follows from
the definition. To see the other inequality, consider that $\mathcal{M}%
_{A^{\prime}B^{\prime}\rightarrow AB}(\rho_{A^{\prime}B^{\prime}%
})=\operatorname{Tr}[\rho_{A^{\prime}B^{\prime}}]\tau_{AB}$, for $\tau_{AB}%
\in\left\{  \tau_{AB}\geq0,E_{N}(\tau_{AB})\leq0\right\}  $, is a particular
completely positive map satisfying $E_{N}(\mathcal{M})=E_{N}(\omega)\leq0$, so
that
\begin{align}
\mathbf{R}(\mathcal{N}) &  =\inf_{\mathcal{M}:E_{N}(\mathcal{M})\leq0}%
\sup_{\psi_{RA}}\mathbf{D}((\operatorname{id}_{R}\otimes\mathcal{N}%
)(\psi_{S_{A}A^{\prime}B^{\prime}S_{B}})\Vert(\operatorname{id}_{R}%
\otimes\mathcal{M})(\psi_{S_{A}A^{\prime}B^{\prime}S_{B}}))\\
&  \leq\inf_{\omega:E_{N}(\omega)\leq0}\mathbf{D}(\psi_{S_{A}S_{B}}%
\otimes\omega_{AB}\Vert\psi_{S_{A}S_{B}}\otimes\tau_{AB})\\
&  =\inf_{\tau_{AB}:E_{N}(\tau_{AB})\leq0}\mathbf{D}(\omega_{AB}\Vert\tau
_{AB})\\
&  =\mathbf{R}(\omega_{AB}).
\end{align}
This concludes the proof.
\end{proof}

\begin{proposition}
[Subadditivity]The max-Rains information of a bipartite channel is subadditive
with respect to serial composition, in the following sense:%
\begin{equation}
R_{\max}(\mathcal{N}_{2}\circ\mathcal{N}_{1})\leq R_{\max}(\mathcal{N}%
_{1})+R_{\max}(\mathcal{N}_{2}).
\end{equation}

\end{proposition}

\begin{proof}
Straightforward and based on methods employed in \cite{WWS19channels}.
\end{proof}

\begin{proposition}
[Faithfulness]The generalized Rains information of a bipartite channel
$\mathcal{N}_{A^{\prime}B^{\prime}\rightarrow AB}$\ obeys the following
faithfulness condition:%
\begin{equation}
\mathbf{R}(\mathcal{N})\geq0\text{ and }\mathbf{R}(\mathcal{N})=0\text{ if and
only if }\mathcal{N}\in\text{\textrm{C-PPT-P},}%
\end{equation}
if the underlying generalized channel divergence obeys the strong faithfulness
condition of \cite{BHKW18}.
\end{proposition}

\begin{proof}
Straightforward and based on methods employed in \cite{WWS19channels}.
\end{proof}

\subsection{Upper bound on distillable entanglement of a bipartite channel}

The three propositions of faithfulness, subadditivity with respect to serial
compositions, and reduction to states leads to a different (perhaps simpler)
proof of the upper bound on distillable entanglement of a bipartite channel, other than that given previously \cite{DBW17}.
Such a protocol has a structure of the following form, preparing a state
$\omega$ at the end%
\begin{equation}
\omega=\mathcal{P}^{n+1}\circ\mathcal{N\circ P}^{n}\circ\cdots\circ
\mathcal{P}^{2}\circ\mathcal{N}\circ\mathcal{P}^{1},
\end{equation}
where the first channel $\mathcal{P}^{1}$\ prepares a PPT\ state. Then it
follows that%
\begin{align}
R_{\max}(\omega)  & =R_{\max}(\mathcal{P}^{n+1}\circ\mathcal{N\circ P}%
^{n}\circ\cdots\circ\mathcal{P}^{2}\circ\mathcal{N}\circ\mathcal{P}^{1})\\
& \leq\sum_{i=1}^{n+1}R_{\max}(\mathcal{P}^{i})+nR_{\max}(\mathcal{N})\\
& =nR_{\max}(\mathcal{N}).
\end{align}
The first equality follows from reduction to states, the inequality from subadditivity, and
the last equality from faithfulness.

The generalized Rains information of a bipartite
channel simplifies to the generalized Rains information of a point-to-point
channel, whenever $\mathcal{N}_{A^{\prime}B^{\prime}\rightarrow AB}$ is a
single-sender, single-receiver channel with trivial $B^{\prime}$ system and
trivial $A$ system. The above then leads to an alternate method of proof of
the main result of \cite{BW17}.

\subsection{$\kappa$-entanglement of bipartite quantum channels}\label{sec:ent-cost}

In this section, we define an entanglement measure $E_{\c}(\mc{N})$ of a bipartite quantum channel $\mc{N}_{A'B'\to AB}$ and show that it is not enhanced by amortization \cite{KW17}, meaning that $E_{\c}(\mc{N})$ is an upper bound on entangling power \cite{BHLS03}.
It is sensible that $E_{\c}(\mc{N})$ is an upper bound on the entanglement cost of a bipartite channel $\mc{N}$ and will be presented in future work. The proof approach follows by adapting to the bipartite setting, the result  from \cite{WW18}.

\begin{definition} The $\kappa$-entanglement 
$E_{\c}(A;B)_{\rho_{AB}}$ of a quantum state $\rho_{AB}$ is defined as  \cite{WW18}
\begin{equation}\label{eq:cost-w}
E_{\c}(A;B)_{\rho}\coloneqq \log W_{\c}(A;B)_{\rho},
\end{equation}
where $W_{\c}(A;B)_{\rho}$ is the solution to the following semidefinite program:
\begin{align}
\textnormal{minimize}\  & \ \Tr\{S_{AB}\}\nonumber\\
\textnormal{subject to}\ &\ S_{AB}\geq 0,\nonumber\\
&\ -S_{AB}^{\T_{B}}\leq \rho_{AB}^{\T_{B}}\leq S_{AB}^{\T_{B}}.
\label{eq:cost-alt}
\end{align}
\end{definition}

The following definition generalizes the $\kappa$-entanglement of a point-to-point channel \cite{WW18} to the bipartite setting.

\begin{definition}\label{def:bi-max-cost} 
The $\kappa$-entanglement $E_{\c}(\mc{N})$  of a bipartite quantum channel $\mc{N}_{A'B'\to AB}$ is defined as
\begin{equation}
E_{\c}(\mc{N})\coloneqq \log \Gamma_{\c} (\mc{N}), \label{eq:bi-max-cost-info}
\end{equation}
where $\Gamma_{\c}(\mc{N})$ is the solution to the following semi-definite program:
\begin{align}
\textnormal{minimize}\ &\ \norm{\Tr_{AB}\{Q_{L_A ABL_B}\}}_\infty\nonumber\\
\textnormal{subject to}\ &\ Q_{L_A ABL_B}\geq 0,\nonumber\\
&\ -Q_{L_A ABL_B}^{\T_{BL_B}}\leq \T_{BL_B}(J^\mc{N}_{L_AABL_B})\leq Q_{L_A ABL_B}^{\T_{BL_B}},
\label{eq:bi-cost-channel-sdp}
\end{align} 
where $L_A\simeq A'$ and $L_B\simeq B'$. 
\end{definition}

\begin{theorem}
[Monotonicity]
Let $\mathcal{M}_{\hat{A}'\hat{B}'\rightarrow\hat{A}\hat{B}
}$ be a bipartite quantum channel and  $\Theta^{\operatorname{PPT}}$ a completely-PPT-preserving superchannel of the form in \eqref{eq:superchannel-action}.
The channel measure $E_\kappa$ is monotone
under the action of the superchannel $\Theta^{\operatorname{PPT}}$, in the sense that%
\begin{equation}
E_\kappa(\mathcal{M}_{\hat{A}'\hat{B}'\rightarrow\hat{A}\hat{B}
}) 
\geq 
E_\kappa(\Theta^{\operatorname{PPT}}(\mathcal{M}_{\hat{A}'\hat{B}'\rightarrow\hat{A}\hat{B}
})).
\end{equation}

\end{theorem}

\begin{proof}
It is a generalization of the related proof given in \cite{WW18} for point-to-point channels. Follows from the definition of $E_\kappa$ and the structure of PPT superchannels.
\end{proof}

\bigskip

%\begin{lemma}
%Consider that $\textbf{C-PPT-P}$ denotes a set of completely PPT-preserving channels. Then $E_{\c}^{2\to 2}(\mc{N})$ can also be expressed as 
%\begin{equation}
%E_{\c}^{2\to 2}(\mc{N}) = \inf_{\mc{M}\in \textbf{C-PPT-P}} D_{\max}(\mc{N}\Vert\mc{M}). 
%\end{equation}
%\end{lemma}
%\begin{proof}
%Consider that $\mc{M}_{A'B'\to AB}$ is a completely PPT-preserving channel. 
%\begin{align}
%& \inf_{\mc{M}\in\textbf{C-PPT-P}} D_{\max}(\mc{N}\Vert\mc{M}) \nonumber\\
%& = \log \inf \{\norm{\T_B\circ\mc{M}_{A'B'\to AB}\circ\T_{B'}}_{\infty}: J^{\mc{N}}\leq J^{\mc{M}} \}\\
%&= \log \inf\{ \norm{\Tr_{AB}(Q_{L_AABL_B})}_{\infty} : Q_{L_AABL_B}\geq 0, \T_{BL_B}(Q_{L_AABL_B})=J^{\mc{M}}_{L_AABL_B}\}\\
%&= = \log \inf \{\norm{\T_B\circ\mc{M}_{A'B'\to AB}\circ\T_{B'}}_{\infty}: J^{\mc{N}}\leq J^{\mc{M}} \}\\
%&= \log \inf\{ \norm{\Tr_{AB}(Q_{L_AABL_B})}_{\infty} : Q_{L_AABL_B}\geq 0, \T_{BL_B}(Q_{L_AABL_B})\geq J^{\mc{N}}_{L_AABL_B}\}\\
%&= E^{2\to 2}_{\kappa}(\mc{N}).
%\end{align}
%\end{proof}

The following proposition constitutes one of our  technical results, and an immediate corollary of it is that $E_{\c}(\mc{N})$ is an upper bound on the amortized $\kappa$-entanglement of a bipartite channel. 

\begin{proposition}\label{prop:cost-tri-ineq}
Let $\rho_{L_A A'B'L_B}$ be a state and let $\mc{N}_{A'B'\to AB}$ be a bipartite channel. Then 
\begin{equation}
E_{\c}(L_A A; B L_B)_{\omega}\leq E_{\c}(L_A A';B' L_B)_{\rho}+ E_{\c}(\mc{N}),
\end{equation}
where $L_A, L_B$ can be of arbitrary size, $\omega_{L_A AB L_B}=\mc{N}_{A'B'\to AB}(\rho_{L_A A'B'L_B})$ and $E_{\c}(\mc{N})$ is stated in Definition~\ref{def:bi-max-cost}. 
\end{proposition}
\begin{proof}
We adapt the proof steps of \cite[Proposition 1]{BW17} to bipartite setting. By removing logarithms and applying \eqref{eq:cost-w} and \eqref{eq:bi-max-cost-info}, the desired inequality is equivalent to the following one:
\begin{equation}\label{eq:w-omega-ineq}
W_{\c}(L_A A; B L_B)_{\omega}\leq W_{\c}(L_A A';B' L_B)_{\rho}\cdot \Gamma_{\c}(\mc{N}),
\end{equation}
and so we aim to prove this one. Exploiting the identity in \eqref{eq:cost-alt}, we find that 
\begin{equation}
W_{\c}(L_A A'; B' L_B)_{\rho}=\min \Tr\{S_{L_AA'B'L_B}\},
\end{equation}
subject to the constraints 
\begin{align}
S_{L_AA'B'L_B} \geq 0,\\
-S_{L_AA'B'L_B}^{\T_{B'L_B}}\leq \rho_{L_AA'B'L_B}^{\T_{B'L_B}}\leq S_{L_AA'B'L_B}^{\T_{B'L_B}},
\end{align}
while the definition in \eqref{eq:bi-cost-channel-sdp} gives that 
\begin{equation}
\Gamma_{\c}(\mc{N})=\min \norm{\Tr_{AB}\{Q_{L'_A ABL'_B}\}}_\infty,
\end{equation}
subject to the constraints 
\begin{align}
Q_{L'_A ABL'_B} \geq 0,\\
 -Q_{L'_A ABL_B'}^{\T_{BL_B'}}\leq \T_{BL_B'}(J^\mc{N}_{L_A'ABL_B'})\leq Q_{L_A' ABL_B'}^{\T_{BL_B'}},\label{eq:choi-bi-b}
\end{align}
where $L_A'\simeq A'$ and $L_B'\simeq B'$.
The identity in \eqref{eq:cost-alt} implies that the left-hand side of \eqref{eq:w-omega-ineq} is equal to 
\begin{equation}
W_{\c}(L_AA;BL_B)_\omega=\min \Tr\{Y_{L_AABL_B}\},
\end{equation}
subject to the constraints
\begin{align}
Y_{L_AABL_B} \geq 0,\\
-Y_{L_AABL_B}^{\T_{BL_B}}\leq \omega_{L_AABL_B}^{\T_{BL_B}}\leq Y_{L_AABL_B}^{\T_{BL_B}},\label{eq:cost-sdp-channel-ef}.
\end{align}

Once we have these SDP formulations, we can now show that the inequality in \eqref{eq:w-omega-ineq} holds by making appropriate choices for $Y_{L_AABL_B}$.  Let $S_{L_AA'B'L_B}$ and $Q_{L_A'ABL_B'}$ be optimal solutions for $W_{\c}(L_AA';B'L_B)_\rho$ and $\Gamma_{\c}(\mc{N})$, respectively. Let $\ket{\Upsilon}_{L_A'L_B':A'B'}$ be the maximally entangled vector. Choose
\begin{align}
Y_{L_AABL_B}&=\bra{\Upsilon}_{L_A'L_B':A'B'}S_{L_AA'B'L_B}\otimes Q_{L_A'ABL_B'}\ket{\Upsilon}_{L_A'L_B':A'B'}.
\end{align}
The above choice can be thought of as a bipartite generalization of that made in the proof of
\cite[Proposition 12]{WW18} (see also
\cite[Proposition 1]{BW17}), and it can be understood roughly understood as a post-selected teleportation of the optimal operator of $W_{\c}(L_AA';B'L_B)_\rho$ through the optimal operator of $\Gamma_{\c}(\mc{N})$, with the optimal operator of $W_{\c}(L_AA';B'L_B)_{\rho}$ being in correspondence with the Choi operator $J^\mc{N}_{L_A'ABL_B'}$ through \eqref{eq:choi-bi-b}. Then, we have, $Y_{L_AABL_B}\geq 0$, because
\begin{equation}
S_{L_AA'B'L_B}, Q_{L_A'ABL_B'}\geq 0.
\end{equation}
We have 
\begin{align}
Y_{L_AABL_B}
& = \bra{\Upsilon}_{L_A'L_B':A'B'} S_{L_AA'B'L_B}\otimes Q_{L_A'ABL_B'}\ket{\Upsilon}_{L_A'L_B':A'B'}\\
& = \Tr_{L_A'A'B'L_B'}\{\Upsilon_{L_A'L_B':A'B'}S_{L_AA'B'L_B}\otimes Q_{L_A'ABL_B'}\},
\end{align}
which implies
\begin{align}
Y_{L_AABL_B}^{\T_{BL_B}}
& = \T_{BL_B}\!\left[\Tr_{L_A'A'B'L_B'}\left\{\Upsilon_{L_A'L_B':A'B'}S_{L_AA'B'L_B}\otimes Q_{L_A'ABL_B'}\right\}\right]\\
& = \T_{BL_B}\!\left[\Tr_{L_A'A'B'L_B'}\left\{\Upsilon_{L_A'L_B':A'B'}S_{L_AA'B'L_B}\otimes (\T_{L_B'}\circ\T_{L_B'})(Q_{L_A'ABL_B'})\right\}\right]\\
& = \T_{BL_B}\!\left[\Tr_{L_A'A'B'L_B'}\left\{\Upsilon_{L_A'L_B':A'B'}^{T^\dag_{L_B'}}S_{L_AA'B'L_B}\otimes Q_{L_A'ABL_B'}^{\T_{L_B'}}\right\}\right]\\
& = \T_{BL_B}\!\left[\Tr_{L_A'A'B'L_B'}\left\{\Upsilon_{L_A'L_B':A'B'}^{T_{L_B'}}S_{L_AA'B'L_B}\otimes Q_{L_A'ABL_B'}^{\T_{L_B'}}\right\}\right]\\
& = \T_{BL_B}\!\left[\Tr_{L_A'A'B'L_B'}\left\{\Upsilon_{L_A'L_B':A'B'}S_{L_AA'B'L_B}^{T_{B'}}\otimes Q_{L_A'ABL_B'}^{\T_{L_B'}}\right\}\right]\\
& = \left[\Tr_{L_A'A'B'L_B'}\left\{\Upsilon_{L_A'L_B':A'B'}S_{L_AA'B'L_B}^{T_{B'L_B}}\otimes Q_{L_A'ABL_B'}^{\T_{BL_B'}}\right\}\right]\\
& \geq \Tr_{L_A'A'B'L_B'}\left\{\Upsilon_{L_A'L_B':A'B'}\rho^{T_{B'L_B}}_{L_AA'B'L_B}\otimes \T_{BL_B'}\(J^\mc{N}_{L_A'ABL_B'}\)\right\}\\
& = \T_{BL_B}\!\left[\Tr_{L_A'A'B'L_B'}\left\{\Upsilon_{L_A'L_B':A'B'}\rho^{T_{B'}}_{L_AA'B'L_B}\otimes \T_{L_B'}\(J^\mc{N}_{L_A'ABL_B'}\)\right\}\right]\\
& = \T_{BL_B}\!\left[\Tr_{L_A'A'B'L_B'}\left\{\Upsilon^{\T_{L_B'}}_{L_A'L_B':A'B'}\rho_{L_AA'B'L_B}\otimes \T_{L_B'}\(J^\mc{N}_{L_A'ABL_B'}\)\right\}\right]\\
& = \T_{BL_B}\!\left[\Tr_{L_A'A'B'L_B'}\left\{\Upsilon_{L_A'L_B':A'B'}\rho_{L_AA'B'L_B}\otimes \T_{L_B'}^\dag\circ \T_{L_B'}\(J^\mc{N}_{L_A'ABL_B'}\)\right\}\right]\\
&=\T_{BL_B}\!\left[\Tr_{L_A'A'B'L_B'}\left\{\Upsilon_{L_A'L_B':A'B'}\rho_{L_AA'B'L_B}\otimes \(J^\mc{N}_{L_A'ABL_B'}\)\right\}\right]\\
& =\T_{BL_B}\left[\mc{N}_{A'B'\to AB}(\rho_{L_AA'B'L_B})\right] 
\end{align}
In the above, we employed properties of the partial transpose, in particular, the fact that partial transpose is self-adjoint. 

Similarly, we have 
\begin{equation}
-Y_{L_AABL_B}^{\T_{BL_B}}\leq \T_{BL_B}\left[\mc{N}_{A'B'\to AB}(\rho_{L_AA'B'L_B})\right],
\end{equation}
where we use the following constraints: 
\begin{equation}
-S_{L_A A'B'L_B}^{\T_{BL_B}}\leq \rho_{L_AA'B'L_B}^{\T_{B'L_B}},\qquad -Q_{L_A' ABL_B'}^{\T_{BL_B'}}\leq \T_{BL_B'}(J^\mc{N}_{L_A'ABL_B'}).
\end{equation}
Thus, $Y_{L_AABL_B}$ is feasible for $W_{\c}(L_AA;BL_B)_{\omega}$. 
Now, we consider
\begin{align}
\Tr\{Y_{L_AABL_B}\}
& = \Tr\{\bra{\Upsilon}_{L_A'L_B':A'B'}S_{L_AA'B'L_B}\otimes Q_{L_A'ABL_B'}\ket{\Upsilon}_{L_A'L_B':A'B'}\}\\
& = \Tr\{S_{L_AA'B'L_B}Q^{T_{A'B'}}_{A'ABB'}\}\\
& = \Tr\{S_{L_AA'B'L_B}\Tr_{AB}\{Q^{T_{A'B'}}_{A'ABB'}\}\}\\
& \leq \Tr\{S_{L_AA'B'L_B}\}\norm{\Tr_{AB}\{Q_{A'ABB'}^{T_{A'B'}}\}}_\infty\\
& = \Tr\{S_{L_AA'B'L_B}\}\norm{\Tr_{AB}\{Q_{A'ABB'}\}}_\infty\\
&=W_{\c}(L_AA';B'L_B)_\rho\cdot \Gamma_{\c}(\mc{N}).
\end{align}
The inequality is a consequence of H\"{o}lder's inequality \cite{Bha97}. The final equality follows because the spectrum of a positive semi-definite operator is invariant under the action of a full transpose (note, in this case, $\T_{A'B'}$ is the full transpose as it acts on reduced positive semi-definite operators $Q_{A'B'}$).

Therefore, we can infer that our choice of $Y_{L_AABL_B}$ a feasible solution of $W_{\c}(L_AA;BL_B)_{\omega}$ such that \eqref{eq:w-omega-ineq} holds. This concludes our proof.
\end{proof}
\bigskip

An immediate corollary of Proposition~\ref{prop:cost-tri-ineq} is the following:
\begin{corollary}\label{cor:rains-tri-ineq}
The quantity $E_{\c}(\mc{N})$ is an upper bound on the amortized $\kappa$-entanglement of a bipartite channel; i.e., the following inequality holds
\begin{equation}
E_{\c,A}(\mc{N}) \leq E_{\c}(\mc{N}),
\end{equation}
where $E_{\c, A} (\mc{N})$ is the amortized entanglement of a bipartite channel $\mc{N}$, i.e.,
\begin{equation}\label{eq:ent-locc-a}
E_{\c, A} (\mc{N})\coloneqq \sup_{\rho_{L_AA'B'L_B}} \left[E_{\c}(L_AA;BL_B)_{\mc{N}(\rho)}-E_{\c}(L_AA';B'L_B)_{\rho}\right],
\end{equation}
and $L_A,L_B$ are of arbitrary size.
\end{corollary}
\begin{proof}
\iffalse
Consider an input state $\Phi_{L_AA'B'L_B}=\Phi_{L_AA'}\otimes\Phi_{L_BB'}$. We have that
\begin{align}
&\sup_{\rho_{L_AA'B'L_B}} \left[E_{\c}(L_AA;BL_B)_{\mc{N}(\rho)}-E_{\c}(L_AA';B'L_B)_{\rho}\right]\nonumber\\
&\geq \left[E_{\c}(L_AA;BL_B)_{\mc{N}(\Phi)}-E_{\c}(L_AA';B'L_B)_{\Phi}\right]\\
&\geq E_{\c}(L_AA;BL_B)_{\mc{N}(\Phi)}\\
&=\inf_{\sigma\in\PPT(L_AA:BL_B)}D_{\max}(\mc{N}_{A'B'\to AB}(\Phi_{LAA'B'L_B})\Vert \sigma_{L_AABL_B})\\
&\geq  \inf_{\mc{M}\in\textbf{C-PPT-P}} D_{\max}(\mc{N}\Vert \mc{M}) = E_{\kappa}(\mc{N}).
\end{align}
\fi
The inequality $E_{\c,A}(\mc{N})\leq E_{\c}(\mc{N})$ is an immediate consequence of Proposition~\ref{prop:cost-tri-ineq}. Let $\rho_{L_AA'B'L_B}$ denote an arbitrary input state. Then from Proposition~\ref{prop:cost-tri-ineq} 
\begin{equation} \label{eq:r-lower-ineq}
E_{\c}(L_AA;BL_B)_\omega-E_{\c}(L_AA';B'L_B)_\rho\leq E_{\c}(\mc{N}),
\end{equation}
where $\omega_{L_AABL_B}=\mc{N}_{A'B'\to AB}(\rho_{L_AA'B'L_B})$. As the inequality holds for any state $\rho_{L_AA'B'L_B}$, we have $E_{\c,A}(\mc{N})\leq E_{\c}(\mc{N})$. 
%Therefore, we can conclude that $E_{\c,A}(\mc{N})= R_{\max}(\mc{N})$.
\end{proof}

\bibliographystyle{alpha}
\bibliography{Ref}

\end{document}